\documentclass[11pt,a4paper]{article}
\usepackage{amsfonts}
\usepackage{amsmath}
\usepackage{amssymb}
\usepackage{amsthm}
\usepackage[margin=1in,a4paper]{geometry}
\usepackage[sc,center,compact,noindentafter,medium]{titlesec}
\usepackage[onehalfspacing,nodisplayskipstretch]{setspace}
\usepackage{graphicx}
\usepackage{url}
\usepackage{array}
\usepackage{float}
\usepackage{color}
\usepackage{xcolor}
\usepackage[authoryear,round]{natbib}
\definecolor{citationblue}{RGB}{0,51,102}
\usepackage[colorlinks=true,citecolor=citationblue,linkcolor=citationblue,urlcolor=citationblue]{hyperref}

\numberwithin{equation}{section}

\allowdisplaybreaks

\newcommand{\CI}{\operatorname{CI}}

\providecommand{\limfunc}[1]{\operatorname{#1}}
\providecommand{\func}[1]{\operatorname{#1}}
\definecolor{darkgreen}{rgb}{0,0.5,0.3}

\theoremstyle{plain}
\newtheorem{theorem}{\normalfont\scshape Theorem}[section]
\newtheorem{corollary}{\normalfont\scshape Corollary}[section]
\newtheorem{lemma}{\normalfont\scshape Lemma}[section]
\newtheorem{assumption}{\normalfont\scshape Assumption}[section]
\theoremstyle{remark}
\newtheorem{remark}{\normalfont\scshape Remark}[section]
\AtBeginDocument{\setlength{\abovedisplayskip}{10pt plus 2pt minus 7pt}\setlength{\abovedisplayshortskip}{0pt plus 2pt}\setlength{\belowdisplayskip}{10pt plus 2pt minus 7pt}\setlength{\belowdisplayshortskip}{5pt plus 2pt minus 3pt}}

\begin{document}

\title{ }

\begin{center}
{\LARGE \textsc{bootstrapping autoregressive duration models}}

\renewcommand{\thefootnote}{} \footnote{\hspace{-7.2mm} $^{a}$Department of
Economics, University of Bologna, Italy and Department of Economics,
University of Exeter, UK. \newline
$^{b}$Department of Economics, University of Copenhagen, Denmark. \newline
We thank conference participants at Aarhus University (Conference in Honor
of James MacKinnon's 75th Birthday), the Hong Kong University of Science and
Technology (3rd HKUST IAS-SBM Joint Workshop in Financial Econometrics),
Toulouse School of Economics (2026 Financial Econometrics Conference), and
the University of Bologna (2026 NBER--NSF Time Series Conference) for
comments and suggestions. We also thank seminar participants at the Chinese
University of Hong Kong, the University of Oxford, and the University of
Macau. G.~Cavaliere and A.~Rahbek gratefully acknowledge support from the
Independent Research Fund Denmark (DFF Grant 7015-00028) and from the
Italian Ministry of University and Research (PRIN 2020 Grant 2020B2AKFW).
Correspondence to: Giuseppe Cavaliere, Department of Economics, University
of Bologna, email: \textsf{giuseppe.cavaliere@unibo.it}.} %
\addtocounter{footnote}{-1} \renewcommand{\thefootnote}{\arabic{footnote}}

{\normalsize \vspace{0.1cm} }

{\large \textsc{Giuseppe Cavaliere}}$^{a}$, {\large \textsc{Anders Rahbek}}$%
^{b}${\large \textsc{and Frederik Vilandt}}$^{b}${\normalsize \vspace{0.2cm}}

\bigskip

{\normalsize \vspace{0.2cm}}

September 21, 2026{\normalsize \vspace{0.2cm}\vspace{0.2cm} }

\bigskip

\textsc{Abstract}\vspace{-0.15cm}
\end{center}

{\small This paper develops bootstrap methods for likelihood-based inference
in autoregressive conditional duration (ACD) models, where the sample size
is endogenously determined by durations observed over a fixed time span.
This feature fundamentally shapes the asymptotic framework, particularly so
when the durations do not have finite expectation. Building on recent limit
theory for heavy-tailed and integrated ACD processes, we analyze recursive
bootstrap schemes that either fix the time span (yielding a random sample
size) or fix the number of durations (yielding a random time span). We
establish a bootstrap theory for ACD models that links naturally to renewal
theory with random sample sizes. For the fixed-count bootstrap, we prove
first-order validity in the finite-mean and boundary cases and characterize
the random limiting bootstrap distribution in the infinite-mean case.
Although classical bootstrap consistency can fail when the durations have
infinite expectation, we argue that the bootstrap remains valid and yields asymptotically normal
t-statistics. Monte Carlo evidence shows that the proposed methods have good
finite-sample properties in both finite- and infinite-mean settings, and are
robust to distributional misspecification relative to the exponential
likelihood. We conclude with an empirical application to cryptocurrency ETFs.%
}\bigskip

\bigskip


\medskip


\bigskip

\newpage

\section{Introduction}

\label{sec:introduction}

Autoregressive conditional duration (ACD) models are designed for
irregularly spaced event data. Their central object is the waiting time
between consecutive events, for example transaction durations, quote
durations, intervention durations or other inter-arrival times in economics
and finance. In these applications the econometrician typically observes all
events occurring during a given calendar-time window. The number of
observations is therefore not chosen directly: it is generated by the
stochastic duration process itself.

{Any bootstrap algorithm must therefore decide whether to reproduce
the original calendar-time window or the realized number of durations.
Preserving the time span appears natural because it mimics how the original
data are collected. Fixing the number of events instead implies a random
bootstrap time span, which need not be close to the original one. Importantly, this feature does not invalidate bootstrap inference.
We show in particular that the fixed-count bootstrap can deliver valid confidence intervals even when it does not
consistently estimate the unconditional asymptotic distribution of the
estimator.}

The ACD model was introduced by \citet{engleRussell1998} for durations
between financial transactions, and has since become a standard tool for
high-frequency durations and other irregularly spaced event data. Important
extensions and surveys include \citet{engle2000}, \citet{pacurar2008},
\citet{hautsch2012}, \citet*{fernandesMedeirosVeiga2016},
\citet{bhogalVariyam2019}, and \citet{sauloPalSouzaVilaDasilva2025}. This
literature typically emphasizes the multiplicative
conditional-mean structure and its close analogy with autoregressive
conditional heteroskedastic (ARCH, GARCH) and multiplicative error model
(MEM) specifications. This analogy is useful but, as demonstrated here, it
can be misleading for asymptotic inference. In particular, in ACD models the
likelihood is evaluated over a random number of durations observed during a
fixed calendar-time span, whereas in ARCH, GARCH and standard MEM models the
sample size is treated as deterministic and given ex ante.

To fix ideas, let $\{t_{i}\}_{i\geq 0}$ denote event times, with $%
0=t_{0}<t_{1}<t_{2}<\cdots $, and define the durations by $%
x_{i}=t_{i}-t_{i-1}>0$. Over a fixed observation window $[0,T]$, $T>0$, the
observed number of events is 
\begin{equation}
n(T)=\max \{k\geq 0:\sum\nolimits_{i=1}^{k}x_{i}\leq T\}.  \label{eq:defnT}
\end{equation}%
The data are the durations $x_{1},\ldots ,x_{n(T)}$. Thus the sample size $%
n(T)$ is random and jointly determined with the same durations that enter
the likelihood. As mentioned, this differs from the usual ARCH and MEM
settings, where the sample size is treated as deterministic.

The randomness of $n(T)$ is crucial. Recent limit theory for ACD models developed by
\citeauthor*{cavaliereMikoschRahbekVilandt2024}
(\citeyear{cavaliereMikoschRahbekVilandt2024},
\citeyear{cavaliereMikoschRahbekVilandt2025},
\citeyear{cavaliereMikoschRahbekVilandt2026}) shows that, under
strict stationarity of the duration process $\{x_{i}\}$, the rate and the
limiting distribution of likelihood estimators depend on the tail index $%
\kappa $ of the stationary distribution of $x_{i}$, which satisfies%
\begin{equation}
\mathbb{P}(x_{i}>x)\sim c_{\kappa }x^{-\kappa },\qquad x\rightarrow \infty ,
\label{eq:tailindex}
\end{equation}%
for some constant $c_{\kappa }>0$. When $\kappa >1$, durations have finite
mean $\mathbb{E}[x_{i}]<\infty $, and the maximum likelihood estimator
converges at the calendar-time rate $\sqrt{T}$. At the boundary $\kappa =1$,
where $\mathbb{E}[x_{i}]=\infty $, the rate is slower, namely $\sqrt{T/\log T%
}$. Finally, when $0<\kappa <1$, durations have infinite mean $\mathbb{E}%
[x_{i}]=\infty $ and the estimator converges at rate $\sqrt{T^{\kappa }}$
and has a mixed normal limiting distribution. These regimes are driven by
the large-sample behavior of the random count $n(T)$, which depends on $%
\kappa $ being less than, equal to, or greater than one.

{We consider two recursive bootstrap schemes.}

A first natural bootstrap scheme preserves calendar time. It recursively
generates bootstrap durations $\{x_{i}^{\ast }\}$ until their cumulative sum
reaches the original time span $T$, and sets 
\begin{equation}
n^{\ast }(T)=\max \{k\geq 0:\sum\nolimits_{i=1}^{k}x_{i}^{\ast }\leq T\}.
\label{eq:bootnT}
\end{equation}%
We label this the \emph{random-count} bootstrap. It mimics how the original
data are collected, but its bootstrap count can be very different from the
observed $n(T)$, particularly when $\kappa $ approaches or falls below one.

A second scheme preserves the realized count and sets the number of events
in the bootstrap world as $n^{\ast }=n(T)$. We label this the \emph{%
fixed-count} bootstrap. The implied bootstrap time span, $T^{\ast
}=\sum\nolimits_{i=1}^{n(T)}x_{i}^{\ast }$, need not equal the original $T$,
{which may seem counterintuitive for point-process data. This scheme
is close to what is often implemented in empirical MEM and ACD applications.
Its validity, however, requires an asymptotic theory that accounts for the
randomness of the observed count.}

{The main contribution of the paper is to provide an asymptotic theory for
the fixed-count bootstrap. We show that this bootstrap consistently estimates
the distribution of the original likelihood estimator when $\kappa \geq 1$.
When $0<\kappa <1$, it does not consistently estimate the unconditional
mixed-normal distribution, but instead reproduces a conditional Gaussian
component of the mixture. Despite this fact, we show that (reverse) percentile bootstrap intervals remain first-order valid,
and that bootstrap $t$-statistics are asymptotically standard normal for all $\kappa >0$. 
In a simulation study, where we compare coverage and average length under
the two bootstrap schemes, we find that fixing the count can
give substantially shorter percentile bootstrap intervals, while keeping the coverage probabilities comparable. 
Moreover, we find that the fixed-count bootstrap performs well even when the innovations are not exponentially distributed.}

The results in the present paper relate to bootstrap inference for point
processes. \citet*{cavaliereLuRahbekStaerkOstergaard2023} develop
bootstrap inference for Hawkes and general point processes under finite-mean
conditions. While their parametric recursive bootstrap corresponds to the
parametric random count bootstrap (see Section~6 in
\citealp{cavaliereLuRahbekStaerkOstergaard2023}),
their theory does not apply unless $\mathbb{E}[x_{i}]<\infty $. Hence, our
problem is different because the ACD recursion can generate durations with
infinite expectation. In this case, the counting process grows more slowly
than calendar time, and the random sample size induces additional randomness
to the asymptotic distribution of the likelihood estimator. This is
precisely the feature highlighted in the ACD (non-bootstrap) limit theory of
\citeauthor*{cavaliereMikoschRahbekVilandt2024}
(\citeyear{cavaliereMikoschRahbekVilandt2024},
\citeyear{cavaliereMikoschRahbekVilandt2025},
\citeyear{cavaliereMikoschRahbekVilandt2026}). In particular,
\citet*{cavaliereMikoschRahbekVilandt2026} show that, for any $\kappa
\leq 1$, the randomness of $n(T)$ breaks the conventional `deterministic
sample size' asymptotics. The present paper takes the next step by asking
which bootstrap schemes remain valid also for infinite-mean durations.

Extant bootstrap literature for the related MEM class is substantial. For
instance, \citet{hidalgoZaffaroni2007} study goodness-of-fit testing for ARCH%
$(\infty )$ models, while \citet*{pereraHidalgoSilvapulle2016} develop a
bootstrap goodness-of-fit test for a class of ACD models.
\citet{pereraSilvapulle2023} provide bootstrap specification tests for dynamic
conditional distribution models. These contributions are directly relevant
to implementation, and their algorithms are close in spirit to the
fixed-count ({nonparametric}) bootstrap procedures considered below. However, their
asymptotic arguments are formulated for a deterministic number of
observations, and they therefore do not address the renewal component
generated by \eqref{eq:defnT}, nor the change in rate and limiting law that
appears when durations have different tail indexes $\kappa $.

The paper is organized as follows. Section~\ref{sec:model} introduces the
model, assumptions and non-bootstrap limit theory. Section~\ref%
{sec:bootstrap} defines the bootstrap schemes and discusses their
implementation, while Section~\ref{sec:bootstrap-theory} states the
bootstrap validity results. Section~\ref{sec:mc} presents results from a
Monte Carlo study. An empirical illustration based on cryptocurrency ETFs is
provided in Section~\ref{sec:emp}. Section~\ref{sec:conclusion} concludes.
The Appendix contains the full proofs and the required auxiliary bootstrap
renewal results.

\medskip

\noindent \textsc{Notation}. We write `$\overset{d}{\rightarrow }$', `$%
\overset{p}{\rightarrow }$' and `$\overset{a.s.}{\rightarrow }$' for
convergence in distribution, probability and almost surely. Conditional
bootstrap probability and expectation are denoted by $\mathbb{P}^{\ast }$
and $\mathbb{E}^{\ast }$. Bootstrap convergence is written as `$\overset{%
d^{\ast }}{\rightarrow }_{a.s.}$', `$\overset{d^{\ast }}{\rightarrow }_{p}$'
or `$\overset{d^{\ast }}{\rightarrow }_{d}$', according to whether the
conditional distribution converges almost surely, in probability or weakly
as a random probability measure \citep[see][]{cavaliereGeorgiev2020}. The
notation $\mathcal{MN}$ denotes a mixed-normal distribution.

\section{Model and non-bootstrap asymptotics}

\label{sec:model}

In this section we summarize the asymptotic (non-bootstrap) limit theory for
ACD\ as given in \citeauthor*{cavaliereMikoschRahbekVilandt2024}
(\citeyear{cavaliereMikoschRahbekVilandt2024},
\citeyear{cavaliereMikoschRahbekVilandt2025},
\citeyear{cavaliereMikoschRahbekVilandt2026}).
Consider the simple exponential ACD model of order one, 
\begin{align}
x_{i}& =\psi _{i}(\theta )\varepsilon _{i},\qquad i=1,2,\ldots ,n(T)
\label{eq-ACD_1} \\
\psi _{i}(\theta )& =\omega +\alpha x_{i-1},  \label{eq-ACD_2}
\end{align}%
where $\theta =(\omega ,\alpha )^{\prime }$, $\omega >0$, $\alpha >0$, and $%
\{\varepsilon _{i}\}$ is i.i.d. exponential with unit mean. The exponential
likelihood is the reference likelihood throughout the paper. In the Monte
Carlo designs below we allow for non-exponential innovations as a robustness
exercise; in that case {the estimator considered here should be interpreted as a quasi-maximum likelihood estimator. Note also that the simple ACD specification in \eqref{eq-ACD_1}%
-\eqref{eq-ACD_2} is used here to keep the arguments transparent; see also Remark~\ref{rem:acd11-literature} below.}

We let $\theta _{0}=\left( \omega _{0},\alpha
_{0}\right) ^{\prime }$ denote the true value of $\theta $, $\omega _{0}>0$, 
$\alpha _{0}>0$ and make the following assumption throughout:

\begin{assumption}
\label{ass:model} (i) The stationarity condition $\mathbb{E}[\log (\alpha
_{0}\varepsilon _{i})]<0$ holds; (ii) The parameter space $\Theta \subset
(0,\infty )^{2}$ is compact and $\theta _{0}$ is an interior point of $%
\Theta $.
\end{assumption}

The MLE\ $\hat{\theta}_{n(T)}=\limfunc{arg\,max}_{\theta \in \Theta }%
\mathcal{L}_{n(T)}(\theta )$, with log-likelihood function given by 
\begin{equation}
\mathcal{L}_{n(T)}(\theta )=\sum_{i=1}^{n(T)}\ell _{i}(\theta ),\qquad \ell
_{i}(\theta )=-\left( \log \psi _{i}(\theta )+\frac{x_{i}}{\psi _{i}(\theta )%
}\right) \text{.}  \label{eq:-log-Lik}
\end{equation}%
For the scalar null hypothesis $\mathsf{H}_{0}:\alpha =\alpha _{0}$, we also
define the studentized statistic%
\begin{equation}
\tau _{n(T)}=\frac{\hat{\alpha}_{n(T)}-\alpha _{0}}{\hat{\sigma}(\hat{\alpha}%
_{n(T)})},\qquad \hat{\sigma}^{2}(\hat{\alpha}_{n(T)})=\iota _{2}^{\prime }%
\mathcal{I}_{n(T)}(\hat{\theta}_{n(T)})^{-1}\iota _{2},  \label{eq:-tstat}
\end{equation}%
where $\iota _{2}=(0,1)^{\prime }$ and the observed information matrix $%
\mathcal{I}_{n(T)}(\hat{\theta}_{n(T)})$ is given by 
\begin{equation}
\mathcal{I}_{n(T)}(\hat{\theta}_{n(T)})=-\left. \frac{\partial ^{2}\mathcal{L%
}_{n(T)}(\theta )}{\partial \theta \partial \theta ^{\prime }}\right\vert
_{\theta =\hat{\theta}_{n(T)}}\text{.}  \label{eq:info}
\end{equation}%
The stationarity condition $\mathbb{E}[\log (\alpha _{0}\varepsilon _{i})]<0$
in Assumption (\ref{ass:model}) corresponds, for exponentially distributed
innovations, to $\alpha _{0}<\exp (\gamma )\simeq 1.78$, where $\gamma $ is
Euler's constant. The tail index $\kappa $ of the stationary solution is the
unique positive solution to 
\begin{equation}
\mathbb{E}[(\alpha _{0}\varepsilon _{i})^{\kappa }]=1,\qquad \text{or
equivalently}\qquad \alpha _{0}=[\Gamma (\kappa +1)]^{-1/\kappa }.
\label{eq:kappa-alpha}
\end{equation}%
Thus $\kappa >1$ if $\alpha _{0}<1$; $\kappa =1$ if $\alpha _{0}=1$; and $%
0<\kappa <1$ if $\alpha _{0}>1$; see
\citet*{cavaliereMikoschRahbekVilandt2024}. Moreover, let 
\begin{equation}
g_{\kappa }(T)=%
\begin{cases}
T, & \kappa >1, \\ 
T/\log T, & \kappa =1, \\ 
T^{\kappa }, & 0<\kappa <1,%
\end{cases}
\label{eq:gT}
\end{equation}%
and define 
\begin{equation}
\Omega =\mathbb{E}\left[ -\frac{\partial ^{2}\ell _{i}(\theta _{0})}{%
\partial \theta \partial \theta ^{\prime }}\right] =\mathbb{E}\left[ \frac{%
v_{i}v_{i}^{\prime }}{\psi _{i}(\theta _{0})^{2}}\right] ,\qquad
v_{i}=(1,x_{i-1})^{\prime }.  \label{eq:Omega}
\end{equation}%
For later use, the associated counting process theory of
\citet{cavaliereMikoschRahbekVilandt2026} shows that%
\begin{eqnarray}
\frac{n(T)}{g_{\kappa }(T)} &=&\frac{n(T)}{T}\overset{\text{a.s.}}{%
\rightarrow }\mu _{0}^{-1}\text{ for }\kappa >1  \notag \\
\frac{n(T)}{g_{\kappa }(T)} &=&\frac{n(T)}{T/\log T}\overset{p}{\rightarrow }%
c_{0}^{-1}\text{ for }\kappa =1  \label{eq g(T) def} \\
\frac{n(T)}{g_{\kappa }(T)} &=&\frac{n(T)}{T^{\kappa }}\overset{d}{%
\rightarrow }\lambda _{\kappa }\text{ for }\kappa <1  \notag
\end{eqnarray}%
where, for $\kappa>1$, $\mu _{0}=\mathbb{E}[x_{i}]=\omega _{0}/(1-\alpha _{0})$; 
$c_{0}=\omega _{0}/\mathbb{E}[\varepsilon _{i}\log \varepsilon _{i}]$ and
$\lambda _{\kappa }^{-1/\kappa }$ is a strictly positive $\kappa $-stable
random variable.

{These results explain why the asymptotic distribution of the likelihood
estimator depends on the tail index. To see this, note that $\mathcal{I}_{n(T)}(\theta _{0})$, the observed information evaluated at
$\theta _{0}$, can be written as
\begin{equation*}
\frac{\mathcal{I}_{n(T)}(\theta _{0})}{g_{\kappa }(T)}
=\frac{n(T)}{g_{\kappa }(T)}
\frac{\mathcal{I}_{n(T)}(\theta _{0})}{n(T)}.
\end{equation*}%
The second term on the right hand side, $\mathcal{I}_{n(T)}(\theta _{0})/n(T)$,
converges almost surely to $\Omega $, irrespective of $\kappa$; see
\citet{cavaliereMikoschRahbekVilandt2026}. 
Thus, the difference between the
three cases is due to the different asymptotic behavior of the number of durations $n(T)$ over the time span $[0,T]$.
For $\kappa \geq 1$, $n(T)/g_{\kappa}(T)$ has a nonzero constant limit, whereas
for $0<\kappa <1$ its limit is random. In the latter case, the normalized
information therefore has the random limit $\lambda _{\kappa }\Omega $.
This randomness is reflected in the mixed Gaussian limit below.}

The asymptotic properties of the MLE\ are presented in the following theorem.

\begin{theorem}[Non-bootstrap limit theory]
\label{thm-main} Under Assumption~\ref{ass:model}, $\hat{\theta}_{n(T)}%
\overset{a.s.}{\rightarrow }\theta _{0}$ as $T\rightarrow \infty $. Moreover:

\begin{enumerate}
\item If $\kappa >1$, then 
\begin{equation*}
\sqrt{T}(\hat{\theta}_{n(T)}-\theta _{0})\overset{d}{\rightarrow }\mathcal{N}%
(0,\mu _{0}\Omega ^{-1}).
\end{equation*}

\item If $\kappa =1$, then 
\begin{equation*}
\sqrt{T/\log T}(\hat{\theta}_{n(T)}-\theta _{0})\overset{d}{\rightarrow }%
\mathcal{N}(0,c_{0}\Omega ^{-1}).
\end{equation*}

\item If $0<\kappa <1$, then 
\begin{equation*}
\sqrt{T^{\kappa }}(\hat{\theta}_{n(T)}-\theta _{0})\overset{d}{\rightarrow }%
\mathcal{MN}(0,\lambda _{\kappa }^{-1}\Omega ^{-1}).
\end{equation*}
\end{enumerate}

\noindent In all three cases, $\tau _{n(T)}\overset{d}{\rightarrow }\mathcal{%
N}(0,1)$.
\end{theorem}

In compact form, the results in Theorem~\ref{thm-main} can be summarized as%
\begin{equation*}
\sqrt{g_{\kappa }(T)}(\hat{\theta}_{n(T)}-\theta _{0})\overset{d}{%
\rightarrow }Y_{\kappa }=\sqrt{A_{\kappa }}Z_{\theta }
\end{equation*}%
with $Z_{\theta }\sim \mathcal{N}_{2}(0,\Omega ^{-1})$ and $A_{\kappa }=\mu
_{0}$ for $\kappa >1$, $A_{\kappa }=c_{0}$ for $\kappa =1$, and $A_{\kappa
}=\lambda _{\kappa }^{-1}$ for $0<\kappa <1$, with $Z_{\theta }$ independent
of $A_{\kappa }$. For inference on $\alpha $, let%
\begin{equation*}
Z_{\alpha }=\iota _{2}^{\prime }Z_{\theta }\sim \mathcal{N}(0,V_{\alpha }),
\qquad V_{\alpha }=\iota _{2}^{\prime }\Omega ^{-1}\iota _{2}.
\end{equation*}%
This notation will be used below.

Theorem~\ref{thm-main} makes clear why bootstrap inference for ACD models is
delicate. The convergence rate of the estimator depends on the tail index,
while the studentized statistic has a standard normal limit in all regimes.
We therefore distinguish between confidence intervals based on studentized
and non-studentized statistics. The result also clarifies to what extent the
ACD limit theory differs from the more familiar GARCH and MEM asymptotics.
For GARCH and MEM likelihoods, the stochastic recurrence generating the data
may be similar, but the estimator is normalized by a deterministic sample
size. Here the rate is inherited from the renewal limit for $n(T)$. This is
why the finite-mean, boundary and infinite-mean cases must be treated
separately even though the likelihood contributions have the same formal
expression.

\begin{remark}[Connection with the ACD$(1,1)$ theory]
\label{rem:acd11-literature} {Although our focus is on the simple ACD(1) specification in \eqref{eq-ACD_1}--\eqref{eq-ACD_2}, the same random-count issues discussed in the paper appear in the ACD$(1,1)$}
model with $\psi _{i}=\omega +\alpha x_{i-1}+\beta \psi _{i-1}$. In that
model the tail index $\kappa $ is determined by $\mathbb{E}[(\alpha
\varepsilon _{i}+\beta )^{\kappa }]=1$, with $\alpha +\beta =1$
corresponding to the integrated ACD boundary ($\kappa =1$).
\citet*{cavaliereMikoschRahbekVilandt2026} show that this boundary has rate $(T/\log
T)^{1/2}$ and a Gaussian limit. Our bootstrap results should be read as the
fixed-count bootstrap counterpart to the non-bootstrap asymptotic theory in
the analytically simplest case.
\end{remark}

\subsection{Confidence intervals}

\label{sec:-CIs}

A confidence interval $\limfunc{CI}_{\tau }$ with nominal coverage $%
100(1-p)\%$ ($p\in \left( 0,1\right) $) can be constructed using the
studentized $\tau _{n(T)}$ in (\ref{eq:-tstat}) as%
\begin{equation}
\func{CI}_{\mathrm{\tau }}=\left[ \hat{\alpha}_{n(T)}-z_{1-p/2}\hat{\sigma}(%
\hat{\alpha}_{n(T)}),\,\hat{\alpha}_{n(T)}-z_{p/2}\hat{\sigma}(\hat{\alpha}%
_{n(T)})\right] ,  \label{eq:CI-II}
\end{equation}%
where $z_{p}$ is the $p$th quantile of the standard normal distribution.

Using the results in Theorem \ref{thm-main} to construct confidence
intervals based on the asymptotic distribution of $\hat{\theta}_{n(T)}$ is
infeasible in practice due to the presence of nuisance parameters and
uncertainty about the true value of $\kappa $. To see this, define the
confidence interval $\func{CI}$ as 
\begin{equation}
\func{CI}=\left[ \hat{\alpha}_{n(T)}-g_{\kappa }(T)^{-1/2}q_{\kappa
}(1-p/2),\,\hat{\alpha}_{n(T)}-g_{\kappa }(T)^{-1/2}q_{\kappa }(p/2)\right] .
\label{eq:CI-I}
\end{equation}%
where $q_{\kappa }(p)$ denotes the $p$th quantile of the asymptotic
distribution of $g_{\kappa }(T)^{1/2}(\hat{\alpha}_{n(T)}-\alpha _{0})$.
This is clearly infeasible; however, its bootstrap counterpart can easily be
constructed, as exemplified in the next section.

\section{Bootstrap algorithms and implementation}

\label{sec:bootstrap}

A simple model-based bootstrap algorithm can be constructed by generating
the durations in the bootstrap world, say $\left\{ x_{i}^{\ast }\right\} $,
recursively as%
\begin{equation}
x_{i}^{\ast }=\psi _{i}^{\ast }(\theta ^{\ast })\varepsilon _{i}^{\ast
},\qquad \psi _{i}^{\ast }(\theta ^{\ast })=\omega ^{\ast }+\alpha ^{\ast
}x_{i-1}^{\ast },\qquad i=1,2,\ldots ,n^{\ast },  \label{eq-RCD}
\end{equation}%
with $x_{0}^{\ast }=x_{0}$ and $n^{\ast }$ the bootstrap sample size. To
construct bootstrap confidence intervals it is convenient to refer to the
unrestricted bootstrap, which sets $\theta ^{\ast }=\hat{\theta}_{n(T)}$.
The bootstrap innovations can be generated either parametrically, setting $%
\varepsilon _{i}^{\ast }\sim \func{Exp}(1)$ and i.i.d. conditionally on the
data, or nonparametrically, by resampling the scaled residuals 
\begin{equation}
\hat{\varepsilon}_{i}^{s}:=\hat{\varepsilon}_{i}/\bar{\varepsilon},\qquad 
\hat{\varepsilon}_{i}=x_{i}/\hat{\psi}_{i},\qquad \hat{\psi}_{i}=\hat{\omega}%
_{n(T)}+\hat{\alpha}_{n(T)}x_{i-1}  \label{eq:scaled-residuals}
\end{equation}%
and $\bar{\varepsilon}:=n(T)^{-1}\sum_{j=1}^{n(T)}\hat{\varepsilon}_{j}$.
Notice that the scaling enforces $\mathbb{E}^{\ast }[\varepsilon _{i}^{\ast
}]=n(T)^{-1}\sum_{i=1}^{n(T)}\hat{\varepsilon}_{i}^{s}=1$ for the bootstrap
shocks, while preserving the non-negativity condition $\varepsilon
_{i}^{\ast }>0$ (a.s.).

\subsection{Random-count and fixed-count bootstrap schemes}

\label{subsec:schemes}

In the context of the ACD model, the sample size $n^{\ast }$ in the
bootstrap world can be defined according to two different schemes. First,
the \emph{random-count} bootstrap keeps the original calendar-time span
fixed and sets 
\begin{equation}
n^{\ast }=n^{\ast }(T)=\max \{k\geq 0:\sum\nolimits_{i=1}^{k}x_{i}^{\ast
}\leq T\}.  \label{eq:random-n-star}
\end{equation}%
Second, the \emph{fixed-count} bootstrap keeps the realized number of
durations fixed and sets 
\begin{equation}
n^{\ast }=n(T).  \label{eq:fixed-n-star}
\end{equation}%
The fixed-count scheme is the focus of the theory below. This choice is
close to the practice in the ACD and MEM bootstrap literature, where the
observed number of durations is typically held fixed. The main difference is
that here $n(T)$ is not a user-chosen deterministic quantity; rather, it is
itself a statistic based on the original durations. Consequently, existing
proofs of bootstrap validity that treat the sample size in the bootstrap
world as deterministic cannot be employed in the current setting; one also
needs to show that replacing $n$ by the random count $n(T)$ preserves the
relevant likelihood expansions and asymptotic properties. The auxiliary
renewal lemmas in the Appendix are used exactly for this purpose.

Define the bootstrap log-likelihood function 
\begin{equation}
\mathcal{L}_{n^{\ast }}^{\ast }(\theta )=\sum_{i=1}^{n^{\ast }}\ell
_{i}^{\ast }(\theta ),\qquad \ell _{i}^{\ast }(\theta )=-\left( \log \psi
_{i}^{\ast }(\theta )+\frac{x_{i}^{\ast }}{\psi _{i}^{\ast }(\theta )}%
\right) ,  \label{eq:boot-likelihood}
\end{equation}%
where $\psi _{i}^{\ast }(\theta )=\omega +\alpha x_{i-1}^{\ast }$. The
bootstrap MLE is 
\begin{equation}
\hat{\theta}_{n^{\ast }}^{\ast }=\limfunc{arg\,max}_{\theta \in \Theta }%
\mathcal{L}_{n^{\ast }}^{\ast }(\theta ).  \label{eq:boot-mle}
\end{equation}%
The associated bootstrap observed information and bootstrap $t$-statistic
are 
\begin{align}
\mathcal{I}_{n^{\ast }}^{\ast }(\hat{\theta}_{n^{\ast }}^{\ast })& =-\left. 
\frac{\partial ^{2}\mathcal{L}_{n^{\ast }}^{\ast }(\theta )}{\partial \theta
\partial \theta ^{\prime }}\right\vert _{\theta =\hat{\theta}_{n^{\ast
}}^{\ast }},  \label{eq:-bootstrap-info} \\
\tau _{n^{\ast }}^{\ast }& =\frac{\hat{\alpha}_{n^{\ast }}^{\ast }-\hat{%
\alpha}_{n(T)}}{\hat{\sigma}(\hat{\alpha}_{n^{\ast }}^{\ast })},\qquad \hat{%
\sigma}^{2}(\hat{\alpha}_{n^{\ast }}^{\ast })=\iota _{2}^{\prime }\mathcal{I}%
_{n^{\ast }}^{\ast }(\hat{\theta}_{n^{\ast }}^{\ast })^{-1}\iota _{2}.
\label{eq:boot-tstat}
\end{align}

\subsection{Bootstrap confidence intervals and tests}

\label{sec:-bootstrap-CIs}

The studentized bootstrap confidence interval equivalent of $\limfunc{CI}%
_{\tau }$ is given by 
\begin{equation}
\func{CI}_{\tau }^{\ast }=\left[ \hat{\alpha}_{n(T)}-q_{\tau ^{\ast }}^{\ast
}(1-p/2)\hat{\sigma}(\hat{\alpha}_{n(T)}),\,\hat{\alpha}_{n(T)}-q_{\tau
^{\ast }}^{\ast }(p/2)\hat{\sigma}(\hat{\alpha}_{n(T)})\right] ,
\label{eq:CI-IIstar}
\end{equation}%
where $q_{\tau ^{\ast }}^{\ast }(p)$ is the empirical $p$th quantile of the
bootstrap studentized statistics $\tau _{n(T)}^{\ast }$.

The (reverse) percentile bootstrap interval $\limfunc{CI}^{\ast }$ corresponding to the infeasible
non-studentized, naive interval $\limfunc{CI}$ is given by%
\begin{equation}
\func{CI}^{\ast }=\left[ 2\hat{\alpha}_{n(T)}-q_{\hat{\alpha}^{\ast }}^{\ast
}(1-p/2),\,2\hat{\alpha}_{n(T)}-q_{\hat{\alpha}^{\ast }}^{\ast }(p/2)\right]
,  \label{eq:CI-Istar}
\end{equation}%
where $q_{\hat{\alpha}^{\ast }}^{\ast }(p)$ is the empirical $p$th quantile of 
$\hat{\alpha}_{n(T)}^{\ast }$. The formula is equivalently obtained from the
quantiles of $g_{\kappa }(T)^{1/2}(\hat{\alpha}_{n(T)}^{\ast }-\hat{\alpha}%
_{n(T)})$; the factor $g_{\kappa }(T)$ cancels in the endpoint expression.

In terms of testing, if interest is in the scalar null hypothesis $\mathsf{H}%
_{0}:\alpha =\alpha _{0}$, inference can be based on the bootstrap statistic 
$\tau _{n(T)}^{\ast }$, see (\ref{eq-boot-taustarnT}). Alternatively, a
restricted bootstrap algorithm and test can be employed, where the bootstrap
sample is generated recursively as 
\begin{equation*}
x_{i}^{\star }=\psi _{i}^{\star }(\tilde{\theta}_{n(T)})\varepsilon
_{i}^{\star },\qquad \psi _{i}^{\star }(\tilde{\theta}_{n(T)})=\tilde{\omega}%
_{n(T)}+\alpha _{0}x_{i-1}^{\star },\qquad i=1,2,\ldots ,n(T),
\end{equation*}%
with $x_{0}^{\star }=x_{0}$ and $\tilde{\theta}_{n(T)}=(\tilde{\omega}%
_{n(T)},\alpha _{0})^{\prime }$ the restricted estimator of $\theta $,
obtained with the null hypothesis imposed. The restricted bootstrap test
statistic has the form%
\begin{equation*}
\tau _{n(T)}^{\star }=\frac{\hat{\alpha}_{n(T)}^{\star }-\alpha _{0}}{\hat{%
\sigma}(\hat{\alpha}_{n(T)}^{\star })}
\end{equation*}%
where $\hat{\alpha}_{n(T)}^{\star }$ and $\hat{\sigma}(\hat{\alpha}%
_{n(T)}^{\star })$ are the MLE obtained on the bootstrap sample and its
standard error, respectively. The asymptotic validity of the fixed-count
versions of these procedures is established in Section~\ref%
{subsec:bootstrap-inference-validity}.

{Finally, the random-count bootstrap analogues of the intervals
$\limfunc{CI}_{\tau}^{\ast}$ and $\limfunc{CI}^{\ast}$, denoted
$\limfunc{CI}_{\tau,\ast}^{\ast}$ and $\limfunc{CI}_{\ast}^{\ast}$,
are defined as in \eqref{eq:CI-IIstar} and \eqref{eq:CI-Istar},
respectively, by replacing $\tau_{n(T)}^{\ast}$ with
$\tau_{n^{\ast}(T)}^{\ast}$ and $\hat{\alpha}_{n(T)}^{\ast}$ with
$\hat{\alpha}_{n^{\ast}(T)}^{\ast}$.
Similarly, the random-count analogue of the restricted bootstrap test
statistic is obtained by replacing $\hat{\alpha}_{n(T)}^{\star}$ with
$\hat{\alpha}_{n^{\ast}(T)}^{\star}$ in the expression for
$\tau_{n(T)}^{\star}$, yielding $\tau_{n^{\ast}(T)}^{\star}$.}

\section{Bootstrap theory}

\label{sec:bootstrap-theory}

\subsection{Validity of the fixed-count bootstrap}

\label{subsec:validity}

A key difference between the mechanics in the fixed-count bootstrap world
and the original world is that in the former, the time span covered by the
bootstrap durations is random and different from the original time span $%
[0,T]$. Precisely, the time span in the bootstrap world is $[0,T^{\ast }]$
with $T^{\ast }:=\sum\nolimits_{i=1}^{n(T)}x_{i}^{\ast }$. The time span $%
T^{\ast }$ is not measurable with respect to the original data and in
general not expected to be close to $T$, in particular when $\kappa <1$.
Indeed, when $\kappa <1$ this bootstrap is unable to replicate the
asymptotic distribution of the original estimator, as shown in the following
theorem. Despite this, as we argue below, this bootstrap still delivers
valid confidence intervals and hypothesis tests.

{The information decomposition in Section~\ref{sec:model} suggests
why fixing the count can nevertheless be useful for inference. The original
and bootstrap estimation errors can be written as
\begin{eqnarray*}
\sqrt{g_{\kappa }(T)}(\hat{\theta}_{n(T)}-\theta _{0})
&=&\sqrt{\frac{g_{\kappa }(T)}{n(T)}}\,
\sqrt{n(T)}(\hat{\theta}_{n(T)}-\theta _{0}), \\
\sqrt{g_{\kappa }(T)}(\hat{\theta}_{n(T)}^{\ast }-\hat{\theta}_{n(T)})
&=&\sqrt{\frac{g_{\kappa }(T)}{n(T)}}\,
\sqrt{n(T)}(\hat{\theta}_{n(T)}^{\ast }-\hat{\theta}_{n(T)}).
\end{eqnarray*}%
The factor $\sqrt{g_{\kappa }(T)/n(T)}$ is the same in both expressions
and is fixed conditional on the original data. For $0<\kappa <1$, it has
a non-degenerate limiting distribution across samples. The fixed-count bootstrap does not
reproduce this distribution across samples; instead, it keeps $\sqrt{g_{\kappa }(T)/n(T)}$ fixed in the bootstrap world. 
As argued below, this is sufficient for the
validity of the fixed-count bootstrap confidence intervals. Establishing this result requires new arguments, because the bootstrap sample size is itself random under the original probability measure.}

For the asymptotic bootstrap theory below we study convergence of bootstrap
sequences, $Y_{n^{\ast}}^{\ast }$ say, with $n^{\ast }=n(T)$; an example is
when $Y_{n^{\ast}}^{\ast }$ is the bootstrap observed information $\mathcal{I%
}_{n^{\ast}}^{\ast }(\theta)$ in (\ref{eq:-bootstrap-info}). A key challenge
for the bootstrap asymptotic theory is that if $Y_{n}^{\ast }$, indexed by a
deterministic sequence $n$, converges in the standard bootstrap modes, i.e.
in probability or in distribution, in probability, to some $Y$, that is, $%
Y_{n}^{\ast }\overset{p^{\ast }}{\rightarrow }_{p}Y$ or $Y_{n}^{\ast }%
\overset{d^{\ast }}{\rightarrow }_{p}Y$ as $n\rightarrow\infty$, then it
does not necessarily follow that $Y_{n^{\ast}}^{\ast }\overset{p^{\ast }}{%
\rightarrow }_{p}Y$ or $Y_{n^{\ast}}^{\ast }\overset{d^{\ast }}{\rightarrow }%
_{p}Y$, even if $n^{\ast }=n(T)\rightarrow_{\mathrm{a.s.}}\infty$.

The two bootstrap renewal results in Lemmas~\ref{lem-replace-n-by-nT-new}
and~\ref{lem-new-n-and-nstar-conv} of the Appendix address this issue. The
first shows that convergence for a deterministic index transfers to $n(T)$
when the stronger notions $\overset{p^{\ast }}{\rightarrow }_{\mathrm{a.s.}}$
and $\overset{d^{\ast }}{\rightarrow }_{\mathrm{a.s.}}$ are used in the
bootstrap world. The second combines this result with the asymptotic
behavior of the random ratio $n(T)/g_{\kappa}(T)$. These two steps are
central to the proof of the following theorem.

\begin{theorem}[Fixed-count bootstrap]
\label{thm-bootstrap-main} Under Assumption~\ref{ass:model}, consider the
nonparametric fixed-count bootstrap $n^*=n(T)$ with $\theta^*=\hat\theta
_{n(T)}$, where the bootstrap innovations are obtained by resampling the
scaled residuals in \eqref{eq:scaled-residuals}. Then, as $T\to\infty$:

\begin{enumerate}
\item If $\kappa >1$, then 
\begin{equation*}
\sqrt{T}(\hat{\theta}_{n(T)}^{\ast }-\hat{\theta}_{n(T)})\overset{d^{\ast }}{%
\rightarrow }_{a.s.}\mathcal{N}(0,\mu _{0}\Omega ^{-1}).
\end{equation*}

\item If $\kappa =1$, then 
\begin{equation*}
\sqrt{T/\log T}(\hat{\theta}_{n(T)}^{\ast }-\hat{\theta}_{n(T)})\overset{%
d^{\ast }}{\rightarrow }_{p}\mathcal{N}(0,c_{0}\Omega ^{-1}).
\end{equation*}

\item If $0<\kappa <1$, then 
\begin{equation*}
\sqrt{T^{\kappa }}(\hat{\theta}_{n(T)}^{\ast }-\hat{\theta}_{n(T)})\overset{%
d^{\ast }}{\rightarrow }_{d}\mathcal{N}(0,\lambda _{\kappa }^{-1}\Omega
^{-1})\mid \lambda _{\kappa },
\end{equation*}%
where $\mathcal{N}\left( 0,\lambda _{\kappa }^{-1}\Omega ^{-1}\right) $ $|$ $%
\lambda _{\kappa }$ denotes the Gaussian distribution $\mathcal{N}\left(
0,\lambda _{\kappa }^{-1}\Omega ^{-1}\right) $ conditional on $
\lambda _{\kappa }$.
\end{enumerate}

Moreover, for any $\kappa >0$, and with $\tau _{n(T)}^{\ast }$ denoting the
bootstrap studentized $t$-statistic%
\begin{equation}
\tau _{n(T)}^{\ast }=\frac{\hat{\alpha}_{n(T)}^{\ast }-\hat{\alpha}_{n(T)}}{%
\hat{\sigma}(\hat{\alpha}_{n(T)}^{\ast })},\qquad \hat{\sigma}^{2}(\hat{%
\alpha}_{n(T)}^{\ast })=\iota _{2}^{\prime }\mathcal{I}_{n(T)}^{\ast }(\hat{%
\theta}_{n(T)}^{\ast })^{-1}\iota _{2},  \label{eq-boot-taustarnT}
\end{equation}%
it holds that $\tau _{n(T)}^{\ast }\overset{d^{\ast }}{\rightarrow }_{p}%
\mathcal{N}(0,1)$.
\end{theorem}

For $\kappa \geq 1$ the fixed-count bootstrap consistently estimates the
limiting distribution of the estimator. For $0<\kappa <1$, the limiting
bootstrap measure is not mixed Gaussian as in Theorem \ref{thm-main};
rather, it is random in the limit and reproduces a (Gaussian) component of
the limiting mixed normal distribution. As a consequence, classical
bootstrap consistency fails in the infinite-mean case. The failure is
nevertheless benign for the validity of percentile bootstrap
inference. To see this, let 
\begin{equation*}
\mathcal{T}_{n(T)}=g_{\kappa }(T)^{1/2}(\hat{\alpha}_{n(T)}-\alpha
_{0}),\qquad \mathcal{T}_{n(T)}^{\ast }=g_{\kappa }(T)^{1/2}(\hat{\alpha}%
_{n(T)}^{\ast }-\hat{\alpha}_{n(T)}),
\end{equation*}%
denote the normalized original and bootstrap estimators, and let $\hat{F}%
_{n(T)}^{\ast }(u)=\mathbb{P}^{\ast }(\mathcal{T}_{n(T)}^{\ast }\leq u)$ be
the distribution function of $\mathcal{T}_{n(T)}^{\ast }$, conditionally on
the data. The one-sided bootstrap p-value is 
\begin{equation}
\hat{p}_{n(T)}^{\ast }=\hat{F}_{n(T)}^{\ast }(\mathcal{T}_{n(T)})\text{,}
\label{eq:boot-pvalue}
\end{equation}%
and we have the following corollary; see also the general bootstrap theory
for similar considerations in \citet[Theorem~3.1]{cavaliereGeorgiev2020}.

\begin{corollary}
\label{cor:boot-pvalue} Under the assumptions of Theorem~\ref%
{thm-bootstrap-main}, $\hat{p}_{n(T)}^{\ast }\overset{d}{\rightarrow }U[0,1]$
as $T\rightarrow \infty $.
\end{corollary}

\subsection{Validity of bootstrap confidence intervals and tests}

\label{subsec:bootstrap-inference-validity}

The validity of the fixed-count procedures described in Section~\ref%
{sec:-bootstrap-CIs} follows from the preceding results. The standard-normal
limits for the original statistic $\tau _{n(T)}$ in Theorem~\ref{thm-main}
and for the bootstrap statistic $\tau _{n(T)}^{\ast }$ in Theorem~\ref%
{thm-bootstrap-main} imply 
\begin{equation*}
\mathbb{P}\left( \alpha _{0}\in \func{CI}_{\tau }^{\ast }\right) \rightarrow
1-p.
\end{equation*}%
Thus, the studentized bootstrap interval in (\ref{eq:CI-IIstar}) has
asymptotically correct coverage. Corollary~\ref{cor:boot-pvalue} establishes
the validity of bootstrap tests based on $\hat{p}_{n(T)}^{\ast }$ and, by
inversion, of the percentile interval in (\ref{eq:CI-Istar}); see,
e.g., \citet[Remark~3.3]{cavaliereGeorgiev2020}.

The same likelihood expansion used in the proof of Theorem~\ref{thm-bootstrap-main} also yields, under the null hypothesis, that the
restricted bootstrap statistic $\tau _{n(T)}^{\star }$ defined in Section~%
\ref{sec:-bootstrap-CIs} satisfies $\tau _{n(T)}^{\star }\overset{d^{\star }}%
{\rightarrow }_{p}\mathcal{N}(0,1)$. Hence, a test based on the conditional
quantiles of $\tau _{n(T)}^{\star }$ has asymptotic size $p$. Although
Theorem~\ref{thm-bootstrap-main} is stated for the nonparametric bootstrap,
the same conclusions hold for the parametric exponential bootstrap by the
argument given in its proof.

\subsection{Non-exponential innovations}

In the case where $\varepsilon _{i}$ are non-exponential with $\mathbb{E}%
\left[ \varepsilon _{i}\right] =1$ and $\sigma _{\varepsilon }^{2}=\mathbb{V}%
\left[ \varepsilon _{i}\right] <\infty $, the $\tau _{n(T)}$ statistic is
asymptotically $\mathcal{N}\left( 0,\sigma _{\varepsilon }^{2}\right) $, cf.
\citet{cavaliereMikoschRahbekVilandt2026}. More specifically, the asymptotic distribution of
the estimator is scaled by $\sigma _{\varepsilon }$ for $\kappa =1$ and $%
\kappa >1$ (we conjecture that the same result holds when $\kappa <1$ as
well), and standard inference becomes invalid. Robust versions of the $\tau _{n(T)}$ statistic, based on standard sandwich estimators of the asymptotic variance, are well known and can in
principle be employed. The simulation evidence in the next section indicates
that the {nonparametric} bootstrap can remain effective under substantial departures
from exponential innovations. Establishing its validity under general
non-exponential innovations, particularly for $0<\kappa<1$, is beyond the
scope of the present paper.

\section{Monte Carlo simulations}

\label{sec:mc}

In this section we provide a Monte Carlo study to evaluate the finite-sample
properties of the proposed fixed-count bootstrap. We also aim at providing
some comparisons with the random-count bootstrap, where the number of
observations is random in the bootstrap world. 
{The main quantities of interest are empirical coverage probabilities,
average interval lengths, and empirical rejection probabilities. The
comparison of coverage and length is of particular interest, as our theory establishes first-order validity of the fixed-count bootstrap intervals, but does not by itself establish their finite-sample performance relative to the random-count bootstrap intervals.}

\subsection{Design}

\begin{table}[t]
\caption{\textsc{Simulation values of $\protect\alpha _0$ from 
\eqref{eq-alpha-of-s-and-kappa}}}
\label{table-alpha-calibration}\centering
\vspace{0.5em} 
\begin{tabular}{cccc}
\hline
& \multicolumn{3}{c}{$s$} \\ \cline{2-4}
$\kappa _0$ & $2.1$ & $3.0$ & $\infty$ \\ \hline
$0.5$ & $1.59$ & $1.44$ & $1.27$ \\ 
$1.0$ & $1.00$ & $1.00$ & $1.00$ \\ 
$1.1$ & $0.91$ & $0.93$ & $0.96$ \\ \hline
\end{tabular}%
\par
\ \vspace{-0.5em}
\end{table}

The baseline data-generating process $\left\{ x_{i}\right\} _{i=1}^{n(T)}$
is \eqref{eq-ACD_1}-\eqref{eq-ACD_2} with $\omega _{0}=1$, and where, for a
given target duration tail index $\kappa $, the value of $\alpha _{0}$ is
chosen from 
\begin{equation}
\mathbb{E}[(\alpha _{0}\varepsilon _{i})^{\kappa }]=1\text{.}
\label{eq:alpha-general}
\end{equation}%
For a given observation window $[0,T]$, $T>0$, durations $%
\{x_{i}\}_{i=1}^{n(T)}$ are simulated using a `burn-in' sample of $d=1000$
observations,%
\begin{equation*}
x_{i}=(\omega _{0}+\alpha _{0}x_{i-1})\varepsilon _{i},\quad i=-(d-1),\dots
,-1,0,1,\dots ,n(T),
\end{equation*}%
with $\varepsilon _{i}$ i.i.d. and initial value $x_{-d}=0$. To examine
robustness to different innovation tails, in addition to $\varepsilon _{i}$
being exponentially distributed, we consider $\varepsilon _{i}$ being Lomax
(Pareto Type II) distributed. The Lomax distribution function, under the
constraint $\mathbb{E}\left[ \varepsilon _{i}\right] =1$, is given by%
\begin{equation*}
F_{\varepsilon ,s}(x)=1-\left( 1+\frac{x}{s-1}\right) ^{-s},\qquad x>0,s>1,
\end{equation*}%
with the exponential distribution obtained as the limiting case by letting $%
s\rightarrow \infty $. This distribution has a right power-law tail with
tail index $s>1$, 
\begin{equation*}
1-F_{\varepsilon ,s}\left( x\right) \sim \left( s-1\right) ^{s}x^{-s},\quad
x\rightarrow \infty .
\end{equation*}%
Thus, $\sigma _{\varepsilon }^{2}(s)=\mathbb{V}\left[ \varepsilon _{i}\right]
<\infty $ provided $s>2$, in which case $\sigma _{\varepsilon }^{2}(s)=%
\mathbb{E}[\varepsilon _{i}^{2}]-1=s/(s-2)$.

In the simulations, we vary the number of finite moments for $\varepsilon
_{i}$ by considering shape parameters as given by $s\in \{2.1,3,\infty \}$,
where $\sigma _{\varepsilon }^{2}\left( 2.1\right) =21$, $\sigma
_{\varepsilon }^{2}\left( 3\right) =3$, and $\sigma _{\varepsilon
}^{2}\left( \infty \right) =1$. In the non-exponential case, the values of $%
\alpha _{0}$ that correspond to particular target values of $\kappa \in
\{0.5,1.0,1.1\}$ are obtained as the unique solution to (\ref%
{eq:alpha-general}) which for any $\kappa <s$ is given by 
\begin{equation}
\alpha _{0}=\left( \mathbb{E}[\varepsilon _{i}^{\kappa }]\right) ^{-1/\kappa
}=\frac{1}{s-1}\left( \frac{\Gamma (\kappa +1)\Gamma (s-\kappa )}{\Gamma (s)}%
\right) ^{-1/\kappa }\text{,}  \label{eq-alpha-of-s-and-kappa}
\end{equation}%
with $\Gamma \left( \cdot \right) $ denoting the Gamma function. Note that
for the exponential case this reduces to \eqref{eq:kappa-alpha}. Table \ref%
{table-alpha-calibration} reports the corresponding values of $\alpha _{0}$
used in the simulations.

\begin{remark}
\label{Tail-of-xi} Note that, using
\citet[Lemma~2.1]{cavaliereMikoschRahbekVilandt2024}, it follows that the duration $x_{i}$, when generated with
Lomax distributed $\varepsilon _{i}$ has (unconditional) power law right
tail with index $\kappa $ solving (\ref{eq:alpha-general}). More precisely,%
\begin{equation*}
\mathbb{P}\left( x_{i}>x\right) \sim c_{\kappa }x^{-\kappa },\quad
x\rightarrow \infty ,
\end{equation*}%
where the tail constant $c_{\kappa }$ is given in (A.2) of
\citet*{cavaliereMikoschRahbekVilandt2024}.
\end{remark}

The observation-period span $T$ is calibrated so that the median event count
across replications, $\limfunc{med}n\left( T\right) $, belongs to $%
\{200,400,800,1600\}$. For each Monte Carlo replication the unrestricted
estimator $\hat{\theta}_{n(T)}$ is computed, and next $B$ bootstrap samples
are generated from the random-count and the fixed-count bootstrap
algorithms, respectively. The reported confidence interval results focus on $%
\alpha$ and compare the random-count and fixed-count versions of the percentile bootstrap interval $\func{CI}^{\ast }$, the asymptotic studentized interval $\func{CI}_{\tau }$, and the random-count and fixed-count versions
of the bootstrap-$t$ interval $\func{CI}_{\tau }^{\ast }$; see Sections~\ref{sec:-CIs} and~\ref{sec:-bootstrap-CIs}.

We also consider tests of the scalar null hypothesis $\mathsf{H}_{0}:\alpha
=\alpha _{0}$. The non-bootstrap test uses 
\begin{equation*}
\tau =(\hat{\alpha}_{n(T)}-\alpha _{0})/\hat{\sigma}(\hat{\alpha}_{n(T)})
\end{equation*}
and compares $|\tau |$ with the standard normal critical value $q=1.96$. The
bootstrap tests are based on a restricted bootstrap generated under the null
hypothesis. Specifically, the bootstrap durations are generated {using $\alpha _{0}$ and $\omega ^{\ast }=\tilde{\omega}_{n(T)}$}, where $%
\tilde{\omega}_{n(T)}$ is the restricted estimator computed from the
original data. The bootstrap innovations are resampled from the scaled
restricted residuals. {With `$^{\star }$' denoting estimates obtained on the restricted bootstrap sample, the bootstrap statistic is} 
{
\begin{equation*}
\tau_{n(T)} ^{\star }=(\hat{\alpha }_{n(T)}^{\star }-\alpha _{0})/\hat{\sigma} (%
\hat{\alpha}_{n(T)}^{\star }),
\end{equation*}
}
where the bootstrap standard error is computed from the bootstrap observed
information as 
{
\begin{equation*}
\hat{\sigma}^{2}(\hat{\alpha}_{n(T)}^{\star })=\iota _{2}^{\prime }%
\mathcal{I}_{n(T)}^{\star }(\hat{\theta}_{n(T)}^{\star })^{-1}\iota
_{2},
\end{equation*}
}
with $\iota _{2}=(0,1)^{\prime }$, as in~\eqref{eq:boot-tstat}. {For the random-count version, the recursively generated count $n^{\ast
}(T)$ defined in~\eqref{eq:random-n-star} replaces $n(T)$.} 

For all designs, we consider $%
M=10000$ Monte Carlo replications and $B=399$ bootstrap replications.

\subsection{Simulation results}

\begin{figure}[t]
\centering
\includegraphics[width=\textwidth]{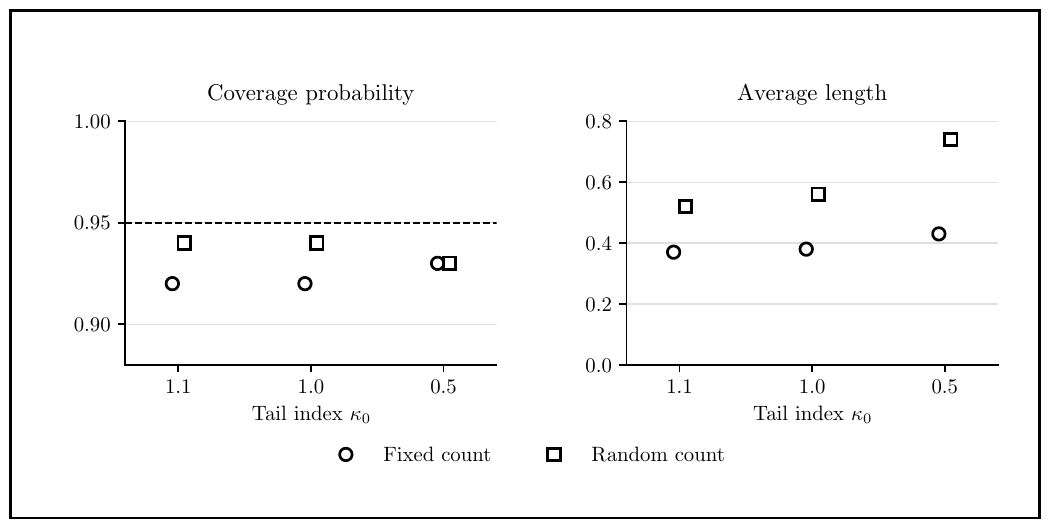}
\smallskip \caption{\textsc{Coverage and average length of reverse percentile bootstrap intervals}}
\label{fig:coverage-length}
\begin{flushleft}
{\small\textit{Notes:} Exponential innovations and
$\limfunc{med}n(T)=400$. The figure reports the coverage probabilities
and average lengths of $\func{CI}^{\ast }$ (fixed count) and
$\func{CI}_{\ast }^{\ast }$ (random count) from
Tables~\ref{table-kappa-eq-1p1}--\ref{table-kappa-eq-0-p5}. The horizontal
line in the left panel marks the nominal coverage $0.95$. Results are
based on $M=10000$ Monte Carlo replications and $B=399$ bootstrap replications.}
\end{flushleft}
\end{figure}

Tables~\ref{table-kappa-eq-1p1}--\ref{table-kappa-eq-0-p5} report empirical
coverage probabilities (CPs) and average lengths (ALs) of confidence
intervals (CIs),
{and Figure~\ref{fig:coverage-length} illustrates the comparison for the case of exponential innovations and $\limfunc{med}n(T)=400$.

For $\kappa _{0}=0.5$, the fixed-count and random-count percentile intervals both have coverage around $0.93$, while their average lengths are $0.43$ and $0.74$, respectively. Thus,
fixing the count gives a substantially shorter interval in this design,
keeping the empirical coverage unchanged. Both procedures slightly undercover the nominal level $0.95$. For $\kappa _{0}=1$ and $1.1$, the fixed-count intervals
are also shorter, but have slightly lower coverage.}

\begin{table}[t]
\caption{\textsc{Coverage probabilities and average lengths of confidence
intervals}}
\label{table-kappa-eq-1p1}\centering{\small \vspace{0.5em}%
\begin{tabular}{llcccccccc}
\hline
\multicolumn{10}{c}{Tail index $\kappa _{0}=1.1$} \\ \hline
& $\limfunc{med}n(T)$ & \multicolumn{2}{c}{$200$} & \multicolumn{2}{c}{$400$}
& \multicolumn{2}{c}{$800$} & \multicolumn{2}{c}{$1600$} \\ \cline{3-10}
\multicolumn{1}{c}{$s$} & \multicolumn{1}{c}{} & CP & AL & CP & AL & CP & AL
& CP & AL \\ \hline
$\infty $ & $\limfunc{CI}_{\tau }$ & $0.96$ & $0.58$ & $0.96$ & $0.39$ & $%
0.95$ & $0.27$ & $0.95$ & $0.19$ \\ 
& $\limfunc{CI}_{\tau }^{\ast }$ & $0.92$ & $0.55$ & $0.92$ & $0.38$ & $0.93$
& $0.26$ & $0.93$ & $0.18$ \\ 
& $\limfunc{CI}_{\tau ,\ast }^{\ast }$ & $0.89$ & $0.49$ & $0.90$ & $0.35$ & 
$0.92$ & $0.25$ & $0.94$ & $0.18$ \\ 
& $\limfunc{CI}^{\ast }$ & $0.92$ & $0.53$ & $0.92$ & $0.37$ & $0.93$ & $%
0.26 $ & $0.93$ & $0.18$ \\ 
& $\limfunc{CI}_{\ast }^{\ast }$ & $0.94$ & $0.72$ & $0.94$ & $0.52$ & $0.95$
& $0.36$ & $0.95$ & $0.24$ \\ \hline
$3$ & $\limfunc{CI}_{\tau }$ & $0.79$ & $0.71$ & $0.78$ & $0.51$ & $0.77$ & $%
0.39$ & $0.76$ & $0.27$ \\ 
& $\limfunc{CI}_{\tau }^{\ast }$ & $0.91$ & $1.03$ & $0.92$ & $0.71$ & $0.93$
& $0.50$ & $0.93$ & $0.35$ \\ 
& $\limfunc{CI}_{\tau ,\ast }^{\ast }$ & $0.87$ & $0.90$ & $0.89$ & $0.64$ & 
$0.91$ & $0.46$ & $0.93$ & $0.33$ \\ 
& $\limfunc{CI}^{\ast }$ & $0.90$ & $0.88$ & $0.92$ & $0.65$ & $0.94$ & $%
0.47 $ & $0.93$ & $0.34$ \\ 
& $\limfunc{CI}_{\ast }^{\ast }$ & $0.90$ & $1.26$ & $0.91$ & $0.95$ & $0.93$
& $0.68$ & $0.94$ & $0.46$ \\ \hline
$2.1$ & $\limfunc{CI}_{\tau }$ & $0.66$ & $0.79$ & $0.63$ & $0.64$ & $0.61$
& $0.39$ & $0.58$ & $0.24$ \\ 
& $\limfunc{CI}_{\tau }^{\ast }$ & $0.90$ & $1.65$ & $0.90$ & $1.10$ & $0.92$
& $0.80$ & $0.92$ & $0.57$ \\ 
& $\limfunc{CI}_{\tau ,\ast }^{\ast }$ & $0.85$ & $1.40$ & $0.87$ & $0.99$ & 
$0.89$ & $0.72$ & $0.91$ & $0.53$ \\ 
& $\limfunc{CI}^{\ast }$ & $0.86$ & $1.15$ & $0.89$ & $0.91$ & $0.91$ & $%
0.70 $ & $0.92$ & $0.53$ \\ 
& $\limfunc{CI}_{\ast }^{\ast }$ & $0.84$ & $1.66$ & $0.87$ & $1.36$ & $0.90$
& $1.04$ & $0.91$ & $0.75$ \\ \hline
\end{tabular}
\begin{flushleft}
{\small \textit{Notes:} The table reports empirical coverage probabilities
(CP) and average length (AL) of confidence intervals. $\CI^{\ast}$ is for the
fixed-count bootstrap in (\ref{eq:CI-Istar}), while $\CI_{\ast}^{\ast}$ is
(\ref{eq:CI-Istar}) for the random-count bootstrap where $n(T)$ is replaced by
$n^{\ast}(T)$. Likewise, $\CI_{\tau}$ is given in (\ref{eq:CI-II}),
$\CI_{\tau}^{\ast}$ is for the fixed-count bootstrap in (\ref{eq:CI-IIstar}),
while $\CI_{\tau,\ast}^{\ast}$ is (\ref{eq:CI-IIstar}) for the random-count
bootstrap where $n(T)$ is replaced by $n^{\ast}(T)$. All intervals are reported
for a $p=0.05$ nominal level. Results are based on $M=10000$ Monte Carlo
replications and, for each replication, $B=399$ bootstrap replications.}
\end{flushleft}
}
\end{table}

\begin{table}[]
\caption{\textsc{Coverage probabilities and average lengths of confidence
intervals}}
\label{table-kappa-eq-1p0}\centering{\small \vspace{0.5em}%
\begin{tabular}{llcccccccc}
\hline
\multicolumn{10}{c}{Tail index $\kappa _{0}=1.0$} \\ \hline
& $\limfunc{med}n(T)$ & \multicolumn{2}{c}{$200$} & \multicolumn{2}{c}{$400$}
& \multicolumn{2}{c}{$800$} & \multicolumn{2}{c}{$1600$} \\ \cline{3-10}
\multicolumn{1}{c}{$s$} & \multicolumn{1}{c}{} & CP & AL & CP & AL & CP & AL
& CP & AL \\ \hline
$\infty $ & $\limfunc{CI}_{\tau }$ & $0.96$ & $0.59$ & $0.96$ & $0.41$ & $%
0.95$ & $0.28$ & $0.95$ & $0.20$ \\ 
& $\limfunc{CI}_{\tau }^{\ast }$ & $0.92$ & $0.56$ & $0.92$ & $0.39$ & $0.93$
& $0.27$ & $0.93$ & $0.19$ \\ 
& $\limfunc{CI}_{\tau ,\ast }^{\ast }$ & $0.88$ & $0.50$ & $0.90$ & $0.36$ & 
$0.91$ & $0.26$ & $0.93$ & $0.18$ \\ 
& $\limfunc{CI}^{\ast }$ & $0.92$ & $0.54$ & $0.92$ & $0.38$ & $0.93$ & $%
0.27 $ & $0.93$ & $0.19$ \\ 
& $\limfunc{CI}_{\ast }^{\ast }$ & $0.93$ & $0.75$ & $0.94$ & $0.56$ & $0.94$
& $0.39$ & $0.95$ & $0.27$ \\ \hline
$3$ & $\limfunc{CI}_{\tau }$ & $0.80$ & $0.80$ & $0.78$ & $0.50$ & $0.77$ & $%
0.34$ & $0.76$ & $0.23$ \\ 
& $\limfunc{CI}_{\tau }^{\ast }$ & $0.91$ & $1.08$ & $0.92$ & $0.74$ & $0.93$
& $0.52$ & $0.93$ & $0.37$ \\ 
& $\limfunc{CI}_{\tau ,\ast }^{\ast }$ & $0.86$ & $0.93$ & $0.89$ & $0.66$ & 
$0.91$ & $0.48$ & $0.92$ & $0.35$ \\ 
& $\limfunc{CI}^{\ast }$ & $0.91$ & $0.91$ & $0.92$ & $0.68$ & $0.94$ & $%
0.49 $ & $0.94$ & $0.35$ \\ 
& $\limfunc{CI}_{\ast }^{\ast }$ & $0.89$ & $1.34$ & $0.91$ & $1.03$ & $0.93$
& $0.74$ & $0.94$ & $0.51$ \\ \hline
$2.1$ & $\limfunc{CI}_{\tau }$ & $0.66$ & $0.94$ & $0.63$ & $0.65$ & $0.61$
& $0.39$ & $0.58$ & $0.26$ \\ 
& $\limfunc{CI}_{\tau }^{\ast }$ & $0.89$ & $1.64$ & $0.90$ & $1.17$ & $0.92$
& $0.84$ & $0.92$ & $0.63$ \\ 
& $\limfunc{CI}_{\tau ,\ast }^{\ast }$ & $0.85$ & $1.43$ & $0.87$ & $1.03$ & 
$0.89$ & $0.75$ & $0.90$ & $0.57$ \\ 
& $\limfunc{CI}^{\ast }$ & $0.87$ & $1.20$ & $0.90$ & $0.96$ & $0.92$ & $%
0.74 $ & $0.92$ & $0.57$ \\ 
& $\limfunc{CI}_{\ast }^{\ast }$ & $0.84$ & $1.79$ & $0.87$ & $1.48$ & $0.89$
& $1.15$ & $0.90$ & $0.85$ \\ \hline
\end{tabular}
\begin{flushleft}
{\small \textit{Notes:} See Table \ref{table-kappa-eq-1p1}.}
\end{flushleft}
}
\end{table}

\begin{table}[t]
\caption{\textsc{Coverage probabilities and average lengths of confidence
intervals}}
\label{table-kappa-eq-0-p5}\centering{\small \vspace{0.5em}%
\begin{tabular}{llcccccccc}
\hline
\multicolumn{10}{c}{Tail index $\kappa _{0}=0.5$} \\ \hline
& $\limfunc{med}n(T)$ & \multicolumn{2}{c}{$200$} & \multicolumn{2}{c}{$400$}
& \multicolumn{2}{c}{$800$} & \multicolumn{2}{c}{$1600$} \\ \cline{3-10}
\multicolumn{1}{c}{$s$} & \multicolumn{1}{c}{} & CP & AL & CP & AL & CP & AL
& CP & AL \\ \hline
$\infty $ & $\limfunc{CI}_{\tau }$ & $0.96$ & $0.71$ & $0.96$ & $0.49$ & $%
0.96$ & $0.34$ & $0.96$ & $0.24$ \\ 
& $\limfunc{CI}_{\tau }^{\ast }$ & $0.92$ & $0.63$ & $0.92$ & $0.46$ & $0.92$
& $0.33$ & $0.92$ & $0.24$ \\ 
& $\limfunc{CI}_{\tau ,\ast }^{\ast }$ & $0.87$ & $0.54$ & $0.88$ & $0.40$ & 
$0.89$ & $0.30$ & $0.91$ & $0.22$ \\ 
& $\limfunc{CI}^{\ast }$ & $0.94$ & $0.59$ & $0.93$ & $0.43$ & $0.93$ & $%
0.31 $ & $0.93$ & $0.22$ \\ 
& $\limfunc{CI}_{\ast }^{\ast }$ & $0.91$ & $0.94$ & $0.93$ & $0.74$ & $0.94$
& $0.57$ & $0.95$ & $0.42$ \\ \hline
$3$ & $\limfunc{CI}_{\tau }$ & $0.82$ & $0.98$ & $0.80$ & $0.62$ & $0.78$ & $%
0.43$ & $0.77$ & $0.30$ \\ 
& $\limfunc{CI}_{\tau }^{\ast }$ & $0.92$ & $1.37$ & $0.92$ & $0.89$ & $0.92$
& $0.67$ & $0.93$ & $0.47$ \\ 
& $\limfunc{CI}_{\tau ,\ast }^{\ast }$ & $0.85$ & $1.23$ & $0.86$ & $0.76$ & 
$0.89$ & $0.58$ & $0.91$ & $0.43$ \\ 
& $\limfunc{CI}^{\ast }$ & $0.94$ & $1.06$ & $0.94$ & $0.80$ & $0.94$ & $%
0.60 $ & $0.94$ & $0.44$ \\ 
& $\limfunc{CI}_{\ast }^{\ast }$ & $0.88$ & $1.71$ & $0.89$ & $1.38$ & $0.91$
& $1.08$ & $0.93$ & $0.81$ \\ \hline
$2.1$ & $\limfunc{CI}_{\tau }$ & $0.68$ & $1.39$ & $0.65$ & $0.87$ & $0.62$
& $0.52$ & $0.60$ & $0.36$ \\ 
& $\limfunc{CI}_{\tau }^{\ast }$ & $0.90$ & $1.94$ & $0.90$ & $1.48$ & $0.91$
& $1.12$ & $0.92$ & $0.85$ \\ 
& $\limfunc{CI}_{\tau ,\ast }^{\ast }$ & $0.82$ & $1.53$ & $0.84$ & $1.21$ & 
$0.87$ & $0.95$ & $0.89$ & $0.74$ \\ 
& $\limfunc{CI}^{\ast }$ & $0.91$ & $1.48$ & $0.92$ & $1.19$ & $0.93$ & $%
0.94 $ & $0.94$ & $0.72$ \\ 
& $\limfunc{CI}_{\ast }^{\ast }$ & $0.82$ & $2.45$ & $0.85$ & $2.09$ & $0.88$
& $1.70$ & $0.89$ & $1.30$ \\ \hline
\end{tabular}
\begin{flushleft}
{\small \textit{Notes:} See Table \ref{table-kappa-eq-1p1}.}
\end{flushleft}
}
\end{table}

{For exponential innovations, $s=\infty $, the asymptotic studentized
interval $\func{CI}_{\tau }$ has coverage close to the nominal level. The
fixed-count bootstrap-$t$ interval $\func{CI}_{\tau }^{\ast }$ has coverage
around $0.92$--$0.93$ across the three values of $\kappa _{0}$, while the
percentile fixed-count interval $\func{CI}^{\ast }$ has coverage between
$0.92$ and $0.94$. The random-count percentile interval
$\func{CI}_{\ast }^{\ast }$ has coverage between $0.91$ and $0.95$ and
is longer in all these designs. The additional length can improve coverage,
but does not do so in every design. The comparison therefore concerns both
coverage and length, rather than length alone.}

The advantages of the residual bootstrap are more visible when the
innovations are non-exponential. For $s=3$, the standard asymptotic interval
under-covers substantially, with CPs typically around $0.76$--$0.82$ even
for the larger samples. In contrast, the fixed-count bootstrap-$t$ intervals
have CPs close to $0.91$--$0.93$, and the percentile fixed-count intervals have
coverage close to the nominal level. For $s=2.1$, where the variance of the
innovations is finite but large, the asymptotic interval performs poorly,
whereas the bootstrap intervals continue to give a substantial correction.
This supports the main practical point of Section~\ref{sec:bootstrap}: the
residual bootstrap can provide useful robustness to misspecification of the
exponential likelihood.

Table~\ref{table-test-alpha} reports empirical rejection probabilities for
the test of $\mathsf{H}_{0}:\alpha =\alpha _{0}$. When $s=\infty $, the
normal critical value gives rejection frequencies close to the nominal
level, whereas the restricted bootstrap tests mildly over-reject. Under
non-exponential innovations, the restricted bootstrap tests substantially
reduce the size distortion. Specifically, rejection frequencies fall from approximately
$0.18$--$0.42$ for the Gaussian test to approximately $0.05$--$0.09$ for the
bootstrap tests. The fixed-count version remains competitive with the
random-count version, while being simpler and avoiding the additional
randomness in the bootstrap event count.

\begin{table}[t]
\caption{\textsc{Empirical rejection probabilities for testing $\protect%
\alpha =\protect\alpha _{0}$}}
\label{table-test-alpha}\centering{\scriptsize \vspace{0.5em} 
\resizebox{\textwidth}{!}{\begin{tabular}{llcccccccccccccc}
\hline
&  & \multicolumn{4}{c}{$\kappa _{0}=1.1$} &  & \multicolumn{4}{c}{$\kappa
_{0}=1.0$} &  & \multicolumn{4}{c}{$\kappa _{0}=0.5$} \\ 
\cline{3-6}\cline{8-11}\cline{13-16}
\multicolumn{1}{c}{$s$} & \multicolumn{1}{c}{$\limfunc{med}n(T)$} & $200$ & $400$ & $800$ & $1600$ &  & $200$ & $400$ & $800$ & $1600$ &  & $200$ & $400$ & $800$ & $1600$ \\ \hline
$\infty $ & $|\tau |>q$ & $0.04$ & $0.04$ & $0.05$ & $0.05$ &  & $0.04$ & $0.04$ & $0.05$ & $0.05$ &  & $0.04$ & $0.04$ & $0.04$ & $0.04$ \\ 
& fixed-count bootstrap & $0.07$ & $0.08$ & $0.07$ & $0.07$ &  & $0.08$ & $0.08$ & $0.07$ & $0.07$ &  & $0.07$ & $0.08$ & $0.08$ & $0.08$ \\ 
& random-count bootstrap & $0.07$ & $0.07$ & $0.06$ & $0.06$ &  & $0.07$ & $0.07$ & $0.06$ & $0.06$ &  & $0.08$ & $0.08$ & $0.07$ & $0.07$ \\ \hline
$3$ & $|\tau |>q$ & $0.21$ & $0.22$ & $0.23$ & $0.24$ &  & $0.20$ & $0.22$ & $0.23$ & $0.24$ &  & $0.18$ & $0.20$ & $0.22$ & $0.23$ \\ 
& fixed-count bootstrap & $0.08$ & $0.07$ & $0.06$ & $0.06$ &  & $0.07$ & $0.07$ & $0.06$ & $0.06$ &  & $0.07$ & $0.07$ & $0.07$ & $0.07$ \\ 
& random-count bootstrap & $0.07$ & $0.06$ & $0.06$ & $0.06$ &  & $0.07$ & $0.07$ & $0.05$ & $0.06$ &  & $0.07$ & $0.07$ & $0.06$ & $0.06$ \\ \hline
$2.1$ & $|\tau |>q$ & $0.34$ & $0.37$ & $0.39$ & $0.42$ &  & $0.34$ & $0.37$ & $0.39$ & $0.42$ &  & $0.32$ & $0.35$ & $0.38$ & $0.40$ \\ 
& fixed-count bootstrap & $0.08$ & $0.08$ & $0.07$ & $0.07$ &  & $0.08$ & $0.08$ & $0.07$ & $0.07$ &  & $0.08$ & $0.08$ & $0.08$ & $0.07$ \\ 
& random-count bootstrap & $0.08$ & $0.07$ & $0.06$ & $0.06$ &  & $0.08$ & $0.07$ & $0.06$ & $0.07$ &  & $0.09$ & $0.08$ & $0.07$ & $0.07$ \\ \hline
\end{tabular}} 
\begin{flushleft}
{\small \textit{Notes:} The table reports empirical rejection probabilities
for the two-sided test of $\mathsf{H}_{0}:\alpha=\alpha _{0}$. The first row
uses $\tau =(\hat{\alpha}-\alpha _{0})/\hat{\sigma}(\hat{\alpha})$ and the
standard normal critical value $q=1.96$. The fixed-count and random-count
bootstrap rows use quantiles of the restricted bootstrap $t$-statistic
generated under $\alpha ^{\ast }=\alpha _{0}$ with scaled restricted
residuals. Results are based on $M=10000$ Monte Carlo replications and, for
each replication, $B=399$ bootstrap replications.}
\end{flushleft}
}
\end{table}

\section{An empirical illustration}

\label{sec:emp}

{We examine whether the choice of the bootstrap sample size matters for inference on the duration persistence using intraday trade
durations for five exchange-traded funds (ETFs) tracking cryptocurrency
prices.} The data are the same as in
\citet*{cavaliereMikoschRahbekVilandt2026}. The ETFs
are the Grayscale Bitcoin Mini Trust (ticker: BTC), Grayscale Ethereum Mini
Trust (ETH), Grayscale Bitcoin Trust (GBTC), Grayscale Ethereum Trust (ETHE)
and Bitwise Bitcoin (BITB). The sample covers 35 trading days starting on
January 2, 2025. During regular trading hours (9:30am to 4:00pm EST), this
gives a raw calendar span $T_{\mathrm{raw}}=35\cdot 23400=819000$ seconds. The corresponding
numbers of observed durations, denoted by $n$, are 19366 for BTC, 35492 for
ETH, 157620 for GBTC, 120104 for ETHE and 51917 for BITB.%

Durations are measured in seconds and are obtained from NASDAQ limit order
book data through the LOBSTER database (\textsf{https://data.lobsterdata.com/%
}). The raw intraday durations are adjusted for deterministic intraday
patterns using cubic splines with knots every 30 minutes, as is standard in
empirical duration analysis; see \citet[Ch.~3]{hautsch2012}.\footnote{%
The processed data can be downloaded at \textsf{%
https://github.com/CMRV-ACD/IACD.}} In the bootstrap implementation below,
the random-count bootstrap preserves the adjusted-duration span
$T_{\mathrm{adj}}=\sum_{i=1}^{n}x_i$, rather than the raw calendar span
$T_{\mathrm{raw}}$.

The empirical specification is the standard ACD(1,1) model 
\begin{equation*}
x_{i}=\psi _{i}\varepsilon _{i},\qquad \psi _{i}=\omega +\alpha
x_{i-1}+\beta \psi _{i-1}.
\end{equation*}
This is the empirically standard extension of the first-order autoregressive
duration model studied in the theory above. We apply the same fixed-count
and random-count {nonparametric} residual bootstrap logic to this specification and report
inference for $\omega $, $\alpha $, $\beta $ and the persistence parameter $%
\alpha +\beta $. As in the simulations, the fixed-count bootstrap sets $%
n^{\ast }=n$, while the random-count bootstrap keeps $T_{\mathrm{adj}}$ fixed.

Figure~\ref{fig:conditional-intensity-day-35} illustrates intraday trading
activity for the five ETF series on a typical trading day, {here the last day in our sample (in the figure, the daily time span is rescaled
to $[0,1]$).} The left panels show the estimated conditional intensity $\hat{\lambda}(t)$.
Under the exponential reference model, this is the piecewise constant
function \citep[see][]{engleRussell1998}
\begin{equation}
\hat{\lambda}(t)=\frac{1}{\hat{\psi}_{n(t-)+1}}
=\frac{1}{\hat{\omega}_{n}+\hat{\alpha}_{n}x_{n(t-)}
+\hat{\beta}_{n}\hat{\psi}_{n(t-)}},\qquad t\in[0,T_{\mathrm{adj}}],
\label{eq:emp-conditional-intensity}
\end{equation}%
where $n(t-)$ denotes the number of events strictly before time $t$, and
the fitted conditional durations $\hat{\psi}_{i}$ are obtained from the
ACD(1,1) recursion using the parameter estimates for the full sample.
The right panels show the corresponding counting process ${n(t)}$ for each
series, recording the cumulative number of trades within the day up to and
including time $t$. Its unit jumps occur at the trading times $\{t_i\}$.

Table~\ref{table-emp-ci} reports unrestricted QML estimates and confidence
intervals. {The BTC results show that the choice of the sample size in the bootstrap world is not innocuous for bootstrap percentile intervals. Specifically, the fixed-count interval for $\alpha +\beta $
is $[1.007,1.023]$, whereas the random-count interval is $[0.969,1.026]$.
The former excludes the integrated boundary $\alpha +\beta =1$, while the
latter includes it.}

% Start this figure on a fresh page so it appears here, at the top.
\clearpage
\begin{figure}[H]
\centering
\includegraphics[width=\textwidth]{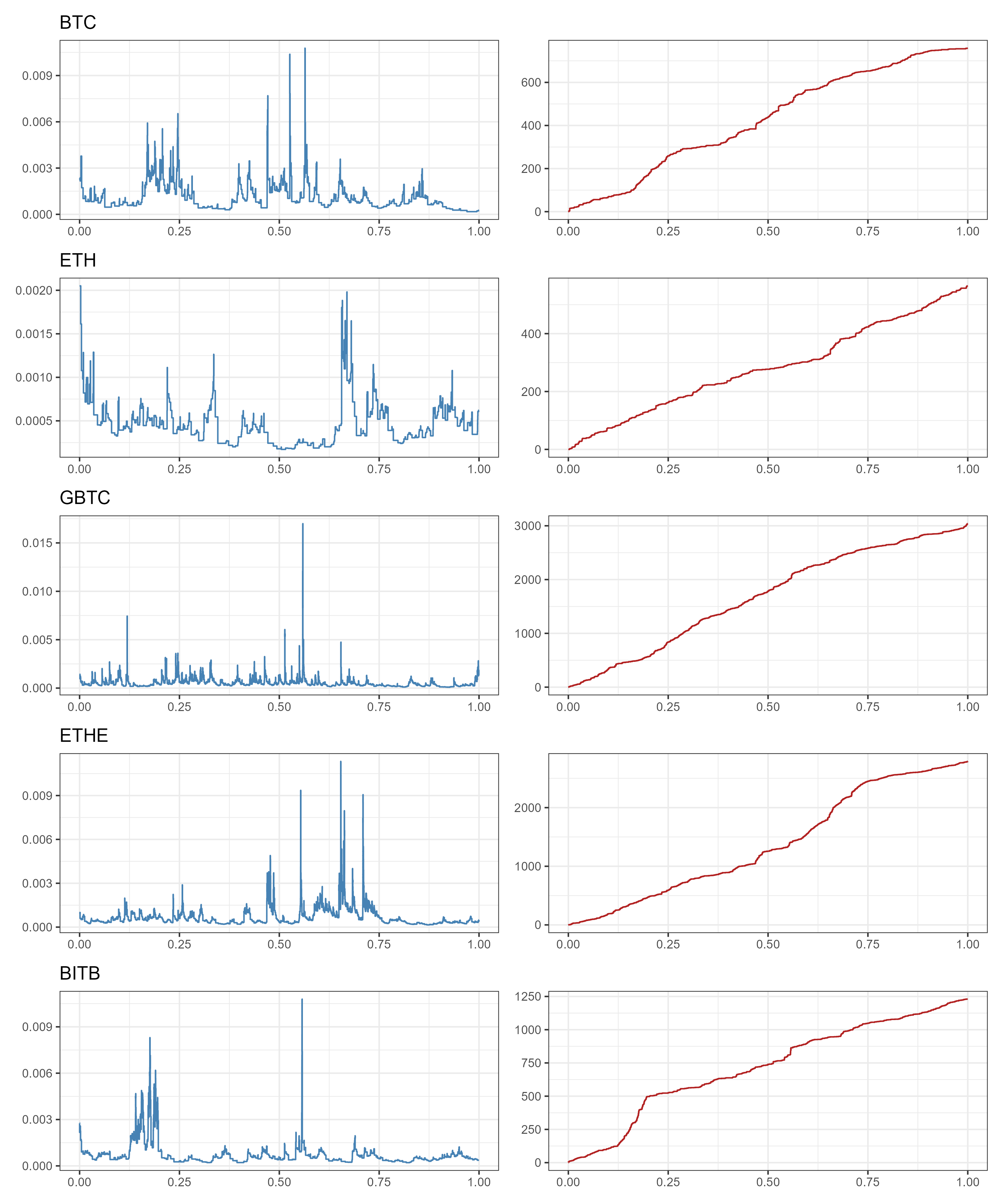}
\vspace{0.5em}
\caption{\textsc{Estimated conditional intensities and counting processes}}
\label{fig:conditional-intensity-day-35}
\begin{flushleft}
{\small\textit{Notes:} Rows correspond to BTC, ETH, GBTC, ETHE and BITB.
Left panels: estimated conditional intensity $\hat{\lambda}(t)$.
Right panels: within-day counting process ${n(t)}$.
Time is reset to zero at the start of the day and rescaled to $[0,1]$.}
\end{flushleft}
\end{figure}

The estimates of $\alpha +\beta $ are above one for all five
ETFs, with values between 1.002 for BITB and 1.018 for ETH. The fixed-count
bootstrap intervals for $\alpha +\beta $ are tight and lie above one for
BTC, ETH, GBTC and ETHE. For BITB, the fixed-count bootstrap intervals
include one, indicating that the evidence against the integrated boundary is
weaker for this series. The random-count percentile intervals are substantially
wider, especially for the intercept, reflecting the additional variation
induced by the random bootstrap event count. For $\omega$, lower endpoints
of confidence intervals are truncated at zero whenever the corresponding
untruncated lower endpoint is negative.

\begin{table}[t]
\caption{\textsc{Unrestricted ACD(1,1) estimates and confidence intervals}}
\label{table-emp-ci}\centering{\scriptsize \vspace{0.5em} 
\begin{tabular}{lcccc}
\hline
& $\omega$ & $\alpha$ & $\beta$ & $\alpha+\beta$ \\ \hline
BTC & 6.663 & 0.186 & 0.829 & 1.015 \\ 
$\func{CI}_{\tau}$ & [6.041, 7.285] & [0.178, 0.194] & [0.823, 0.835] & 
[1.012, 1.019] \\ 
$\func{CI}_{\tau}^{\ast}$ & [5.299, 7.909] & [0.169, 0.204] & [0.815, 0.842]
& [1.009, 1.022] \\ 
$\func{CI}_{\tau,\ast}^{\ast}$ & [5.477, 7.987] & [0.171, 0.203] & [0.814,
0.842] & [1.009, 1.022] \\ 
$\func{CI}^{\ast}$ & [4.704, 8.013] & [0.169, 0.203] & [0.816, 0.843] & 
[1.007, 1.023] \\ 
$\func{CI}_{\ast}^{\ast}$ & [1.284, 9.898] & [0.111, 0.216] & [0.793, 0.873]
& [0.969, 1.026] \\ \hline
ETH & 0.007 & 0.123 & 0.896 & 1.018 \\ 
$\func{CI}_{\tau}$ & [0.005, 0.009] & [0.120, 0.126] & [0.894, 0.898] & 
[1.017, 1.020] \\ 
$\func{CI}_{\tau}^{\ast}$ & [0.000, 0.012] & [0.114, 0.133] & [0.889, 0.901]
& [1.013, 1.023] \\ 
$\func{CI}_{\tau,\ast}^{\ast}$ & [0.000, 0.011] & [0.115, 0.130] & [0.890,
0.900] & [1.014, 1.022] \\ 
$\func{CI}^{\ast}$ & [0.000, 0.011] & [0.111, 0.135] & [0.886, 0.904] & 
[1.013, 1.023] \\ 
$\func{CI}_{\ast}^{\ast}$ & [0.000, 0.014] & [0.047, 0.171] & [0.849, 0.945]
& [0.977, 1.028] \\ \hline
GBTC & 1.974 & 0.119 & 0.896 & 1.015 \\ 
$\func{CI}_{\tau}$ & [1.802, 2.145] & [0.117, 0.121] & [0.894, 0.898] & 
[1.014, 1.016] \\ 
$\func{CI}_{\tau}^{\ast}$ & [1.526, 2.420] & [0.114, 0.125] & [0.892, 0.900]
& [1.013, 1.017] \\ 
$\func{CI}_{\tau,\ast}^{\ast}$ & [1.493, 2.332] & [0.113, 0.124] & [0.892,
0.901] & [1.013, 1.017] \\ 
$\func{CI}^{\ast}$ & [1.682, 2.219] & [0.115, 0.124] & [0.893, 0.899] & 
[1.013, 1.017] \\ 
$\func{CI}_{\ast}^{\ast}$ & [0.000, 2.881] & [0.064, 0.135] & [0.882, 0.920]
& [0.972, 1.021] \\ \hline
ETHE & 0.394 & 0.083 & 0.927 & 1.010 \\ 
$\func{CI}_{\tau}$ & [0.375, 0.414] & [0.081, 0.085] & [0.926, 0.929] & 
[1.010, 1.011] \\ 
$\func{CI}_{\tau}^{\ast}$ & [0.310, 0.467] & [0.076, 0.089] & [0.923, 0.932]
& [1.008, 1.013] \\ 
$\func{CI}_{\tau,\ast}^{\ast}$ & [0.303, 0.454] & [0.078, 0.090] & [0.922,
0.931] & [1.008, 1.013] \\ 
$\func{CI}^{\ast}$ & [0.240, 0.498] & [0.076, 0.089] & [0.923, 0.932] & 
[1.007, 1.013] \\ 
$\func{CI}_{\ast}^{\ast}$ & [0.000, 0.684] & [0.052, 0.102] & [0.908, 0.946]
& [0.997, 1.016] \\ \hline
BITB & 4.836 & 0.095 & 0.906 & 1.002 \\ 
$\func{CI}_{\tau}$ & [4.298, 5.374] & [0.092, 0.099] & [0.903, 0.910] & 
[1.001, 1.003] \\ 
$\func{CI}_{\tau}^{\ast}$ & [3.816, 5.996] & [0.089, 0.102] & [0.900, 0.912]
& [0.999, 1.004] \\ 
$\func{CI}_{\tau,\ast}^{\ast}$ & [3.818, 5.871] & [0.088, 0.103] & [0.899,
0.913] & [0.999, 1.004] \\ 
$\func{CI}^{\ast}$ & [3.901, 5.677] & [0.089, 0.101] & [0.901, 0.911] & 
[0.999, 1.005] \\ 
$\func{CI}_{\ast}^{\ast}$ & [3.191, 6.639] & [0.080, 0.106] & [0.896, 0.916]
& [0.992, 1.005] \\ \hline
\end{tabular}
\begin{flushleft}
{\small \textit{Notes:} Estimates and confidence intervals for the
35-trading-day cryptocurrency ETF sample. The notation follows
Tables~\ref{table-kappa-eq-1p1}--\ref{table-kappa-eq-0-p5}; the subscript
$\ast$ denotes the random-count bootstrap. Intervals are based on $B=399$
bootstrap replications. Confidence intervals for $\omega$ are reported with
lower endpoints truncated at zero.}
\end{flushleft}
}
\end{table}

Table~\ref{table-emp-iacd} reports bootstrap tests of the integrated ACD
hypothesis $\alpha +\beta =1$. The table gives the observed $t$-statistic
and the $2.5\%$ and $97.5\%$ bootstrap quantiles obtained under the
restriction $\alpha ^{\ast }+\beta ^{\ast }=1$. We report quantiles based
both on unrestricted residuals and on restricted residuals. The latter are
used to generate the restricted bootstrap draws and lead to slightly sharper
critical values, but the empirical conclusions are unchanged.

\begin{table}[t]
\caption{\textsc{Bootstrap tests of $\protect\alpha+\protect\beta=1$}}
\label{table-emp-iacd}\centering{\small \vspace{0.5em} 
\resizebox{\textwidth}{!}{\begin{tabular}{lccccc}
\hline
& BTC & ETH & GBTC & ETHE & BITB \\ \hline
$t$ & $9.11$ & $23.24$ & $35.88$ & $33.23$ & $3.11$ \\ 
fixed $n$ & $[-4.12,4.16]$ & $[-6.24,6.52]$ & $[-5.57,4.16]$ & $[-8.12,7.64]$ & $[-4.39,4.04]$ \\ 
random $n$ & $[-4.10,4.39]$ & $[-6.82,6.43]$ & $[-5.34,4.41]$ & $[-8.12,9.05]$ & $[-4.13,3.84]$ \\ 
fixed $n$ (restr. resid.) & $[-4.05,4.05]$ & $[-5.62,5.71]$ & $[-5.36,4.08]$ & $[-7.36,6.54]$ & $[-4.66,4.02]$ \\ 
random $n$ (restr. resid.) & $[-4.06,4.18]$ & $[-5.66,5.67]$ & $[-4.70,4.04]$ & $[-7.25,7.22]$ & $[-3.99,3.99]$ \\ \hline
\end{tabular}} 
\begin{flushleft}
{\small \textit{Notes:} The first row reports
$t=(\hat{\alpha}+\hat{\beta}-1)/\hat{\sigma}(\hat{\alpha}+\hat{\beta})$. The
remaining rows report the $0.025$ and $0.975$ quantiles of the restricted
bootstrap $t$-statistic generated under $\alpha+\beta=1$, using either
unrestricted residuals or scaled restricted residuals. Quantiles are based
on $B=399$ bootstrap replications.}
\end{flushleft}
}
\end{table}

Using the restricted-residual fixed-count bootstrap, the integrated ACD
hypothesis is rejected for BTC, ETH, GBTC and ETHE. For BITB, the observed
statistic $3.11$ lies within the bootstrap quantile interval $[-4.66,4.02]$,
and the integrated specification is therefore not rejected at the 5 percent
level. Since all unrestricted estimates satisfy $\hat{\alpha}+\hat{\beta}%
\geq 1$, we also do not reject the null hypothesis of infinite expected
durations against the finite-mean alternative $\alpha +\beta <1$. Overall,
the empirical illustration points to very persistent, heavy-tailed duration
dynamics in cryptocurrency ETF trading, and it also shows that bootstrap
critical values can matter empirically relative to the standard Gaussian
approximation.

\section{Conclusion}

\label{sec:conclusion}

This paper studies bootstrap inference in ACD models when the number of
observed durations is random. The treatment of the number of events in the
bootstrap world is central. A random-count bootstrap preserves the
calendar-time span, while a fixed-count bootstrap preserves the realized
number of durations. The latter is simple to implement and closely aligned
with existing empirical practice, but its validity cannot be justified by
standard deterministic-sample-size bootstrap arguments.

The fixed-count bootstrap consistently estimates the limiting distribution
for $\kappa \geq 1$. For $0<\kappa <1$, classical bootstrap consistency fails,
but the resulting random limiting measure still yields valid first-order
inference. Specifically, the random limiting measure matches the conditional
Gaussian component of the asymptotic mixed normal distribution, yielding
first-order validity of the percentile bootstrap interval and asymptotic
standard normality of the bootstrap $t$-statistic. This allows valid inference without requiring direct estimation of the limiting mixed normal distribution.

{The main implication is that reproducing the unconditional
asymptotic distribution of the original estimator in the bootstrap world is not necessary for valid inference in this setting. The fixed-count bootstrap retains
the scale determined by the observed number of durations and reproduces one of the components of the limiting mixed normal distribution of the original estimator. We prove theoretically that this property is sufficient for the bootstrap to deliver valid confidence intervals and hypothesis tests. The empirical illustration also shows that the choice of the sample size in the bootstrap world can affect inference about the parameters of interest and, e.g., about the hypothesis of finite expected durations. Thus, the choice between preserving calendar time (random count bootstrap) and preserving the observed count (fixed count bootstrap) has consequences for both the asymptotic justification and the empirical use of bootstrap inference.} This conclusion is not specific to ACD models and may also be relevant for other econometric models in which the sample size is generated endogenously by the underlying data generating process.

\nocite{leeHansen1994}
\bibliographystyle{plainnat}
\bibliography{ACD_Bootstrap_090926_GCAR_v4-1}

\bigskip 

\appendix

\section*{Appendix: Proofs and Auxiliary Results}

\section{Proof of Theorem \protect\ref{thm-main}}

With $n(T)$ replaced by a deterministic sequence, $n=1,2,\ldots$, it follows
from \citet{kristensenRahbek2005} that $\hat{\theta}_{n}\rightarrow_{\text{%
a.s.}}\theta_{0}$. Next, $\hat{\theta}_{n(T)}\rightarrow_{\text{a.s.}%
}\theta_{0}$, holds by using \citet[Theorem~2.1]{gut2009} as $n(T)\rightarrow_{%
\text{a.s.}}\infty$ by Lemma \ref{lem-properties-of-NT}. The distributional
results (i)\ and (iii) are given in
\citeauthor{cavaliereMikoschRahbekVilandt2024}
(\citeyear{cavaliereMikoschRahbekVilandt2024},
\citeyear{cavaliereMikoschRahbekVilandt2025}), while (ii)\ is given in
\citet{cavaliereMikoschRahbekVilandt2026}.\hfill$\square$

\begin{lemma}
\label{lem-properties-of-NT} Under Assumption~\ref{ass:model}(i), $%
n(T)\rightarrow_{\text{a.s.}}\infty$ as $T\rightarrow\infty$.
\end{lemma}

\noindent\textsc{Proof:} For any fixed $M>0$, $\mathbb{P}(n(T)<M)=\mathbb{P}%
(\sum_{i=1}^{M}x_{i}>T)\rightarrow0$ as $T\rightarrow\infty$. Thus $%
n(T)\rightarrow_{p}\infty$ as $T\rightarrow\infty$, and as $n(T)$ is
monotonically increasing, it holds that $n(T)\rightarrow_{\text{a.s}.}\infty$
using, e.g., \citet[Theorem~5.3.5]{gut2013}.\hfill$\square$

\section{Bootstrap theory}

\label{app:-sec-Bootstrap-Theory}

\subsection{Bootstrap renewal results}

\label{app:-subsec-n-to-random-n}

The following lemmas formalize the random-index issue discussed in Section~%
\ref{sec:bootstrap-theory}. We first show that almost sure bootstrap
convergence for a deterministic index is preserved when that index is
replaced by the random count $n(T)$.

\begin{lemma}
\label{lem-replace-n-by-nT-new}Let $\{Y_{n}^{\ast}\}_{n\geq 1}$ be a
sequence of bootstrap random variables indexed by the deterministic integer $%
n$. Under Assumption~\ref{ass:model}(i):

(a) if $Y_{n}^{\ast }\overset{p^{\ast }}{\rightarrow }_{\text{a.s.}}Y$ as $%
n\rightarrow \infty $, then $Y_{n(T)}^{\ast }\overset{p^{\ast }}{\rightarrow 
}_{\text{a.s.}}Y$ as $T\rightarrow \infty ;$

(b) if $Y_{n}^{\ast }\overset{d^{\ast }}{\rightarrow }_{\text{a.s.}}Y$ as $%
n\rightarrow \infty $, then $Y_{n(T)}^{\ast }\overset{d^{\ast }}{\rightarrow 
}_{\text{a.s.}}Y$ as $T\rightarrow \infty $.
\end{lemma}

\noindent\textsc{Proof:} Fix any $\delta>0$ and define the deterministically
indexed $Z_{n}=\mathbb{P}^{\ast}(|Y_{n}^{\ast}-Y|>\delta)$. Under
(a) $Z_{n}\rightarrow _{\text{a.s}.}0$, and as $n(T)\rightarrow_{\text{a.s.}%
}\infty$ by Lemma~\ref{lem-properties-of-NT}, it follows by
\citet[Theorem~2.1]{gut2009} that 
\begin{equation*}
Z_{n(T)}=\mathbb{P}^{\ast}(|Y_{n(T)}^{\ast}-Y|>\delta)\rightarrow _{\text{%
a.s.}}0\text{,}
\end{equation*}
i.e., $Y_{n(T)}^{\ast}\overset{p^{\ast}}{\rightarrow}_{\text{a.s.}}Y$.
Define next the deterministically indexed characteristic function, $\phi
_{n}^{\ast}\left( s\right) =\mathbb{E}^{\ast}[\exp(isY_{n}^{\ast})]$, where 
$s\in\mathbb{R}$. Under (b), $\phi_{n}^{\ast}\left( s\right)
\rightarrow_{\text{a.s.}}\phi(s)$ and hence, as $n(T)\rightarrow _{\text{a.s.%
}}\infty$, $\phi_{n(T)}^{\ast}\left( s\right) \rightarrow _{\text{a.s.}%
}\phi(s)$ as desired. \hfill$\square$

\bigskip

Another important lemma for our theory concerns the relationship between the
count $n(T)$ and the length $T$ of the observation period $[0,T]$. We
formulate our results in terms of the ratio $R_{n(T)}$ as given by%
\begin{equation}
R_{n(T)}=\left( n(T)/g\left( T\right) \right) ^{1/2}  \label{app:-def-ratio}
\end{equation}%
where $g\left( T\right) $ is a deterministic and increasing sequence
satisfying $g\left( T\right) \rightarrow \infty $ as $T\rightarrow \infty $
(see, e.g., eq. (\ref{eq:gT})).

\begin{lemma}
\label{lem-new-n-and-nstar-conv}Under Assumption~\ref{ass:model}(i), assume
that for some bootstrap sequence $\left\{ Y_{n}^{\ast}\right\} _{n=1,2,...}$
indexed by the deterministic $n$, 
\begin{equation}
Y_{n}^{\ast}\overset{d^{\ast}}{\rightarrow}_{\text{a.s.}}Y,\text{ as }%
n\rightarrow\infty\text{,}  \label{Y-Assump-det}
\end{equation}
with $Y$ a scalar random variable with continuous distribution function $F$.
Then, as $T\rightarrow\infty$, with $R_{n(T)}$ defined in (\ref%
{app:-def-ratio}),

(a) if $R_{n(T)}\rightarrow _{\text{a.s}}R>0$, then $Y_{n(T)}^{\ast
}/R_{n(T)}\overset{d^{\ast }}{\rightarrow }_{\text{a.s.}}Y/R$;

(b) if $R_{n(T)}\rightarrow _{p}R>0$, then $Y_{n(T)}^{\ast }/R_{n(T)}\overset%
{d^{\ast }}{\rightarrow }_{p}Y/R$;

(c) if $R_{n(T)}\rightarrow _{d}R>0$, then $Y_{n(T)}^{\ast }/R_{n(T)}\overset%
{d^{\ast }}{\rightarrow }_{d}Y/R$ $|$ $R$.

\noindent where $R$ is a positive constant in (a) and (b), and a positive
random variable in (c). In (c), $Y$ is taken independent of $R$, and $Y/R$ $%
| $ $R$ has the (random) cdf given by $F(uR)$, $u\in\mathbb{R}$.
\end{lemma}

\noindent \textsc{Proof.} By (\ref{Y-Assump-det}), $F_{n}^{\ast }\left(
u\right) =\mathbb{P}^{\ast }(Y_{n}^{\ast }\leq u)\rightarrow _{\text{a.s.}%
}F(u)$, and hence, by continuity of $F$ and the uniform convergence theorem
for distribution functions, $\sup_{u}\left\vert F_{n}^{\ast }\left( u\right)
-F(u)\right\vert \rightarrow _{\text{a.s.}}0$. Thus, by Lemma \ref%
{lem-replace-n-by-nT-new}, $\sup_{u}|F_{n(T)}^{\ast }\left( u\right)
-F(u)|\rightarrow _{\text{a.s.}}0$. For $T$ large, we have%
\begin{equation*}
\mathbb{P}^{\ast }(Y_{n(T)}^{\ast }/R_{n(T)}\leq u)=F_{n(T)}^{\ast
}(uR_{n(T)})=F(uR_{n(T)})+o_{\text{a.s}.}(1)
\end{equation*}%
using $|F_{n(T)}^{\ast }(uR_{n(T)})-F(uR_{n(T)})|\leq \sup_{u\in \mathbb{R}%
}|F_{n(T)}^{\ast }(u)-F(u)|=o_{\text{a.s}.}(1)$. The results in (a) and (b)
follow directly, while for (c) it holds by the continuous mapping theorem, $%
F(uR_{n(T)})\rightarrow _{d}F(uR)$ as desired.\hfill $\square $

\subsection{Proof of Theorem \protect\ref{thm-bootstrap-main}}

The proof is given for the nonparametric residual bootstrap. The proof for
the parametric exponential bootstrap is based on the same arguments and
therefore omitted.

Recall the ACD process $\{x_{i}\}_{i\geq 1}$ in (\ref{eq-ACD_1})--(\ref%
{eq-ACD_2}), the log-likelihood function $\mathcal{L}_{n(T)}(\theta )$ in (%
\ref{eq:-log-Lik}), and the bootstrap analogs $\{x_{i}^{\ast }\}_{i\geq 1}$
and $\mathcal{L}_{n(T)}^{\ast }(\theta )$ defined in (\ref{eq-RCD}) and (\ref%
{eq:boot-likelihood}). In line with the bootstrap renewal theory, let $%
\mathcal{L}_{n}(\theta )$ and $\mathcal{L}_{n}^{\ast }(\theta )$ denote the
original and bootstrap log-likelihood functions indexed by a deterministic $%
n $. In particular, $\hat{\theta}_{n}=\arg \max_{\theta \in \Theta }\mathcal{%
L}_{n}(\theta )$ and $\hat{\theta}_{n}^{\ast }=\arg \max_{\theta \in \Theta }%
\mathcal{L}_{n}^{\ast }(\theta )$.

Corresponding to cases (i)-(iii), let $g\left( T\right) =g_{\kappa }(T)$
defined in (\ref{eq:gT}). With $R_{n(T)}$ defined in (\ref{app:-def-ratio}),
note that 
\begin{equation*}
\sqrt{g(T)}(\hat{\theta}_{n(T)}^{\ast }-\hat{\theta}_{n(T)})=\sqrt{n(T)}(%
\hat{\theta}_{n(T)}^{\ast }-\hat{\theta}_{n(T)})/R_{n(T)}\text{.}
\end{equation*}%
Note that for case (i), $R_{n(T)}\rightarrow _{\text{a.s.}}1/\sqrt{\mu _{0}}$%
, for (ii), $R_{n(T)}\rightarrow _{p}1/\sqrt{c_{0}}$ and for (iii), $%
R_{n(T)}\rightarrow _{d}\sqrt{\lambda _{\kappa }}$. Next, by Lemma \ref%
{lem-asym-normality-boot-of-estimator}, for any $\kappa $,%
\begin{equation}
\sqrt{n(T)}(\hat{\theta}_{n(T)}^{\ast }-\hat{\theta}_{n(T)})\overset{d^{\ast
}}{\rightarrow }_{\text{a.s.}}\mathcal{N}(0,\Omega ^{-1})  \label{app:-next}
\end{equation}%
and hence the desired result follows by Lemma \ref{lem-new-n-and-nstar-conv}%
. The result for the $t$ ratio follows by convergence a.s. of the Hessian as
in Lemma \ref{lem-A7-CPR} and then using Lemma \ref{lem-replace-n-by-nT-new}%
(a).\hfill $\square $

\subsection{Proof of Corollary \protect\ref{cor:boot-pvalue}}

Let $V_{\alpha }=\iota _{2}^{\prime }\Omega ^{-1}\iota _{2}$. By Theorem \ref%
{thm-bootstrap-main}, the conditional cdf of 
\begin{equation*}
\mathcal{T}_{n(T)}^{\ast }=g_{\kappa }(T)^{1/2}(\hat{\alpha}_{n(T)}^{\ast }-%
\hat{\alpha}_{n(T)})
\end{equation*}%
converges to the cdf $F_{A_{\kappa }}$ of $A_{\kappa }^{1/2}Z_{\alpha }$
conditionally on $A_{\kappa }$, where $Z_{\alpha }\sim N(0,V_{\alpha })$ and 
\begin{equation*}
A_{\kappa }=\mu _{0}\quad (\kappa >1),\qquad A_{\kappa }=c_{0}\quad (\kappa
=1),\qquad A_{\kappa }=\lambda _{\kappa }^{-1}\quad (0<\kappa <1).
\end{equation*}%
In the first two cases $A_{\kappa }$ is deterministic. In the infinite-mean
case it is random, but the original statistic and the limiting bootstrap
measure have the same conditional Gaussian law given $\lambda _{\kappa }$.
Proceeding as in the proof of Lemma \ref{lem-new-n-and-nstar-conv}, it
follows for $\kappa \geq 1$ by 
\begin{equation*}
\left\vert \hat{F}_{n(T)}^{\ast }(\mathcal{T}_{n(T)})-F_{A_{\kappa }}(%
\mathcal{T}_{n(T)})\right\vert \leq \sup_{u\in \mathbb{R}}\left\vert \hat{F}%
_{n(T)}^{\ast }(u)-F_{A_{\kappa }}(u)\right\vert \rightarrow _{\text{a.s.}}0,
\end{equation*}%
that $\hat{F}_{n(T)}^{\ast }(\mathcal{T}_{n(T)})\rightarrow F_{A_{\kappa
}}(A_{\kappa }^{1/2}Z_{\alpha })$ in distribution. For $\kappa <1$, $%
F_{A_{\kappa }}$ is the random cdf of $A_{\kappa }^{1/2}Z_{\alpha }$
conditionally on $A_{\kappa }$; the proof follows by replacing $\overset{%
d^{\ast }}{\rightarrow }_{d}$ with $\overset{d^{\ast }}{\rightarrow }_{\text{%
a.s.}}$ using Skorokhod coupling in an extended probability space, as in
\citet[Proof of Theorem~4.1]{cavaliereGeorgiev2020}.

Finally, conditionally on $A_{\kappa }$ when $\kappa <1$, or unconditionally
when $\kappa \geq 1$, the probability integral transform gives $F_{A_{\kappa
}}(A_{\kappa }^{1/2}Z_{\alpha })\sim U[0,1]$ and the unconditional limit is
uniform. This proves the corollary.\hfill $\square $

\begin{lemma}
\label{lem-asym-normality-boot-of-estimator}Under Assumption~\ref{ass:model}%
, as $T\rightarrow\infty$, $\hat{\theta}_{n(T)}^{\ast}-\hat{\theta}_{n(T)}%
\overset{p^{\ast}}{\rightarrow}_{\text{a.s.}}0$. Moreover,%
\begin{equation*}
\sqrt{n(T)}(\hat{\theta}_{n(T)}^{\ast}-\hat{\theta}_{n(T)})\overset{d^{\ast}}%
{\rightarrow}_{\text{a.s.}}\mathcal{N}(0,\Omega^{-1}),
\end{equation*}
where $\Omega=\mathbb{E}[v_{i}v_{i}^{\prime}/\psi_{i}^{2}(\theta_{0})]$ and $%
v_{i}=(1,x_{i-1})^{\prime}$.
\end{lemma}

\noindent \textsc{Proof of Lemma \ref{lem-asym-normality-boot-of-estimator}}%
. Consistency follows by Lemma \ref{lem-replace-n-by-nT-new}, provided $\hat{%
\theta}_{n}^{\ast }-\hat{\theta}_{n}\overset{p^{\ast }}{\rightarrow }_{\text{%
a.s.}}0$ as $n\rightarrow \infty $. We have for any $\delta >0$, and some $%
\eta >0$, 
\begin{align*}
\mathbb{P}^{\ast }\left( ||\hat{\theta}_{n}^{\ast }-\hat{\theta}%
_{n}||>\delta \right) & \leq \mathbb{P}^{\ast }\left( \mathcal{L}_{n}(\hat{%
\theta}_{n})-\mathcal{L}_{n}(\hat{\theta}_{n}^{\ast })>\eta n\right) \\
& =\mathbb{P}^{\ast }\left( \mathcal{L}_{n}(\hat{\theta}_{n})-\mathcal{L}%
_{n}^{\ast }(\hat{\theta}_{n})+\mathcal{L}_{n}^{\ast }(\hat{\theta}_{n})-%
\mathcal{L}_{n}(\hat{\theta}_{n}^{\ast })>\eta n\right) \\
& \leq \mathbb{P}^{\ast }\left( \mathcal{L}_{n}(\hat{\theta}_{n})-\mathcal{L}%
_{n}^{\ast }(\hat{\theta}_{n})+\mathcal{L}_{n}^{\ast }(\hat{\theta}%
_{n}^{\ast })-\mathcal{L}_{n}(\hat{\theta}_{n}^{\ast })>\eta n\right) \\
& \leq \mathbb{P}^{\ast }(2n^{-1}\sup_{\theta \in \Theta }|\mathcal{L}%
_{n}^{\ast }(\theta )-\mathcal{L}_{n}(\theta )|>\eta )\rightarrow _{\text{%
a.s.}}0
\end{align*}%
by Lemma \ref{lem-A2-CPR}. Asymptotic normality holds by Lemma \ref%
{lem-replace-n-by-nT-new} provided $\sqrt{n}(\hat{\theta}_{n}^{\ast }-\hat{%
\theta}_{n})\overset{d^{\ast }}{\rightarrow }_{\text{a.s.}}\mathcal{N}\left(
0,\Omega ^{-1}\right) $ as $n\rightarrow \infty $. The latter holds by
establishing, 
\begin{align*}
\text{(i)}& \text{:}\text{ }\left. \frac{1}{\sqrt{n}}\frac{\partial \mathcal{%
L}_{n}^{\ast }(\theta )}{\partial \theta }\right\vert _{\theta =\hat{\theta}%
_{n}}\overset{d^{\ast }}{\rightarrow }_{\text{a.s.}}\mathcal{N}\left(
0,\Omega \right) \\
\text{(ii)}& \text{:}\text{ }\left. -\frac{1}{n}\frac{\partial ^{2}\mathcal{L%
}_{n}^{\ast }(\theta )}{\partial \theta \partial \theta ^{\prime }}%
\right\vert _{\theta =\hat{\theta}_{n}}\overset{p^{\ast }}{\rightarrow }_{%
\text{a.s.}}\Omega \\
\text{(iii)}& \text{:}\text{ }\max_{l,j,k=1,2}\sup_{\theta \in \Theta
}\left\vert \frac{1}{n}\frac{\partial ^{3}\mathcal{L}_{n}^{\ast }(\theta )}{%
\partial \theta _{l}\partial \theta _{j}\partial \theta _{k}}\right\vert
\leq C_{n}^{\ast },\text{ \ \ }C_{n}^{\ast }\overset{p^{\ast }}{\rightarrow }%
_{\text{a.s.}}C<\infty ,
\end{align*}%
which hold by Lemmas \ref{lem-A5-CPR}, \ref{lem-A7-CPR} and \ref{lem-A9-CPR}%
, respectively. \hfill $\square $

\subsection{Bootstrap likelihood lemmas}

Throughout the following arguments, $C$ denotes a generic finite positive
constant, which may change from line to line.

Recall that $\mathcal{L}_{n}^{\ast }\left( \theta \right) $ and $\mathcal{L}%
_{n}\left( \theta \right) $ denote the bootstrap and original log-likelihood
functions indexed by a deterministic $n$, whereas the corresponding randomly
indexed functions are $\mathcal{L}_{n(T)}^{\ast }\left( \theta \right) $ and 
$\mathcal{L}_{n(T)}\left( \theta \right) $. Lemma \ref{lem-A2-CPR} considers
uniform convergence of $\mathcal{L}_{n}^{\ast }\left( \theta \right) $,
while Lemmas \ref{lem-A5-CPR}, \ref{lem-A7-CPR} and \ref{lem-A9-CPR} provide
results for the derivatives (score, information and third-order derivative).

By Assumption~\ref{ass:model}(ii), there exist constants
$0<\omega _{L}<\omega _{U}<\infty $ and
$0<\alpha _{L}<\alpha _{U}<\infty $ such that
\begin{equation*}
\Theta \subseteq [\omega _{L},\omega _{U}]\times[\alpha _{L},\alpha _{U}],
\end{equation*}%
with $\omega _{L}<\omega _{0}<\omega _{U}$ and
$\alpha _{L}<\alpha _{0}<\alpha _{U}$.

\begin{lemma}
\label{lem-A2-CPR}Let $G_{n}^{\ast}(\theta)=\frac{1}{n}\left[ \mathcal{L}%
_{n}^{\ast}(\theta)-\mathcal{L}_{n}(\theta)\right] $. Under the assumptions
of Lemma \ref{lem-asym-normality-boot-of-estimator}, $\sup_{\theta\in\Theta
}|G_{n}^{\ast}(\theta)|$ $\overset{p^{\ast}}{\rightarrow}_{\text{a.s.}}0$ as 
$n\rightarrow\infty.$
\end{lemma}

\textsc{Proof:} By \citet[Lemma~B.4]{cavaliereNielsenRahbek2017}, it
suffices to show that, as $n\rightarrow\infty$, 
\begin{align*}
\text{(i)} & \text{:}\text{\ }G_{n}^{\ast}(\theta)\overset{p^{\ast }}{%
\rightarrow}_{\text{a.s.}}0,\text{ \ }\theta\in\Theta \\
\text{(ii)} & \text{:}\text{\ }|G_{n}^{\ast}(\theta)-G_{n}^{\ast}(\tilde{%
\theta})|\leq B_{n}^{\ast}||\theta-\tilde{\theta}||,\text{ \ }\tilde{\theta}%
,\theta\in\Theta
\end{align*}
where $B_{n}^{\ast}$ does not depend on $\theta,\tilde{\theta}$ and
satisfies $\mathbb{E}^{\ast}[B_{n}^{\ast}]\rightarrow_{\text{a.s.}}C<\infty$%
. To do so, we repeatedly use that for all $\theta,\tilde{\theta}\in\Theta$, 
\begin{equation}
\psi_{i}^{\ast}(\theta)/\psi_{i}^{\ast}(\tilde{\theta})\leq\omega_{U}/%
\omega_{L}+\alpha_{U}/\alpha_{L}<\infty  \label{eq-bound-on-psi}
\end{equation}

\emph{Establishing (i): }With $Q(\theta )=\mathbb{E}[\log \psi _{i}(\theta
)+x_{i}/\psi _{i}(\theta )]$, by simple insertion%
\begin{equation*}
|G_{n}^{\ast }(\theta )|\leq \underset{|A_{n}\left( \theta \right) |}{%
\underbrace{|-n^{-1}\mathcal{L}_{n}\left( \theta \right) -Q(\theta )|}}+%
\underset{|A_{n}^{\ast }\left( \theta \right) |}{\underbrace{|-n^{-1}%
\mathcal{L}_{n}^{\ast }\left( \theta \right) -Q(\theta )|}}\text{.}
\end{equation*}%
Here $|A_{n}\left( \theta \right) |\rightarrow 0$ a.s., using standard
arguments as in \citet{kristensenRahbek2005}. Next, by definition and using $%
\mathbb{E}\left[ x_{i}/\psi _{i}(\theta )\right] =\mathbb{E}\left[ \psi
_{i}\left( \theta _{0}\right) /\psi _{i}(\theta )\right] $, 
\begin{align*}
A_{n}^{\ast }\left( \theta \right) & =n^{-1}\sum_{i=1}^{n}\log \psi
_{i}^{\ast }(\theta )-\mathbb{E}[\log \psi _{i}(\theta
)]+n^{-1}\sum_{i=1}^{n}\tfrac{x_{i}^{\ast }}{\psi _{i}^{\ast }(\theta )}-%
\mathbb{E}\left[ \tfrac{x_{i}}{\psi _{i}(\theta )}\right] \\
& =\underset{A_{1n}^{\ast }\left( \theta \right) }{\underbrace{%
n^{-1}\sum_{i=1}^{n}\log \psi _{i}^{\ast }(\theta )-\mathbb{E}[\log \psi
_{i}(\theta )]}}+\underset{A_{2n}^{\ast }\left( \theta \right) }{\underbrace{%
n^{-1}\sum_{i=1}^{n}\tfrac{\psi _{i}^{\ast }\left( \theta _{0}\right) }{\psi
_{i}^{\ast }(\theta )}-\mathbb{E}\left[ \tfrac{\psi _{i}\left( \theta
_{0}\right) }{\psi _{i}(\theta )}\right] }} \\
& +\underset{A_{3n}^{\ast }\left( \theta \right) }{\underbrace{%
n^{-1}\sum_{i=1}^{n}(\varepsilon _{i}^{\ast }-1)\tfrac{\psi _{i}^{\ast }(%
\hat{\theta}_{n})}{\psi _{i}^{\ast }(\theta )}}}+\underset{A_{4n}^{\ast
}\left( \theta \right) }{\underbrace{n^{-1}\sum_{i=1}^{n}\tfrac{\psi
_{i}^{\ast }(\hat{\theta}_{n})-\psi _{i}^{\ast }(\theta _{0})}{\psi
_{i}^{\ast }(\theta )}}}
\end{align*}%
Here $A_{1n}^{\ast }\left( \theta \right) \overset{p^{\ast }}{\rightarrow }_{%
\text{a.s.}}0$ by Lemma \ref{recursive-bootstrap-LLN} with $f\left( x\right)
=\log \left( \omega +\alpha x\right) $. Similarly, $A_{2n}^{\ast }\left(
\theta \right) \overset{p^{\ast }}{\rightarrow }_{\text{a.s.}}0$ by Lemma %
\ref{recursive-bootstrap-LLN} with $f(x)=(\omega _{0}+\alpha _{0}x)/(\omega
+\alpha x)$. Next, by Chebyshev's (conditional) inequality 
\begin{align*}
\mathbb{P}^{\ast }\left( \left\vert A_{3n}^{\ast }\left( \theta \right)
\right\vert >\delta \right) & \leq \delta ^{-2}\mathbb{V}^{\ast
}[\varepsilon _{i}^{\ast }]n^{-2}\sum_{i=1}^{n}\mathbb{E}^{\ast }\left[
\left( \tfrac{\psi _{i}^{\ast }(\hat{\theta}_{n})}{\psi _{i}^{\ast }(\theta )%
}\right) ^{2}\right] \\
& \leq C\delta ^{-2}\mathbb{V}^{\ast }[\varepsilon _{i}^{\ast
}]n^{-1}\rightarrow _{\text{a.s.}}0\text{,}
\end{align*}%
using $\psi _{i}^{\ast }(\hat{\theta}_{n})/\psi _{i}^{\ast }(\theta )\leq C$
and $\mathbb{V}^{\ast }[\varepsilon _{i}^{\ast }]=O_{\text{a.s.}}(1)$ as $%
T\rightarrow \infty $ by Lemma \ref{lem-A11-CPR}. Finally, using $\psi
_{i}^{\ast }(\hat{\theta}_{n})=\hat{\omega}_{n}+\hat{\alpha}%
_{n}x_{i-1}^{\ast }$,%
\begin{align*}
\mathbb{P}^{\ast }\left( \left\vert A_{4n}^{\ast }\left( \theta \right)
\right\vert >\delta \right) & =\mathbb{P}^{\ast }\left( \left\vert
n^{-1}\sum_{i=1}^{n}\tfrac{\hat{\omega}_{n}-\omega _{0}+(\hat{\alpha}%
_{n}-\alpha _{0})x_{i-1}^{\ast }}{\omega +\alpha x_{i-1}^{\ast }}\right\vert
>\delta \right) \\
& \leq \mathbb{I}\left( \frac{|\hat{\omega}_{n}-\omega _{0}|}{\omega _{L}}+%
\frac{|\hat{\alpha}_{n}-\alpha _{0}|}{\alpha _{L}}>\delta \right)
\rightarrow _{\text{a.s.}}0,
\end{align*}%
which concludes establishing (i).

\emph{Establishing (ii)}: For the equicontinuity condition, note that 
\begin{equation*}
|G_{n}^{\ast }(\theta )-G_{n}^{\ast }(\tilde{\theta})|\leq \underset{%
|B_{n}^{\ast }\left( \theta \right) |}{\underbrace{\left\vert n^{-1}\mathcal{%
L}_{n}^{\ast }(\theta )-n^{-1}\mathcal{L}_{n}^{\ast }(\tilde{\theta}%
)\right\vert }}+\underset{|B_{n}\left( \theta \right) |}{\underbrace{%
\left\vert n^{-1}\mathcal{L}_{n}\left( \theta \right) -n^{-1}\mathcal{L}_{n}(%
\tilde{\theta})\right\vert }}.
\end{equation*}%
For the first term, with $v_{i}^{\ast }=(1,x_{i-1}^{\ast })^{\prime }$ and $%
\bar{\theta}_{i}$ a point on the line segment between $\theta $ and $\tilde{%
\theta}$, possibly depending on $i$, we have by mean value expansion%
\begin{align*}
|B_{n}^{\ast }\left( \theta \right) |& \leq n^{-1}\sum_{i=1}^{n}|\log \psi
_{i}^{\ast }(\theta )-\log \psi _{i}^{\ast }(\tilde{\theta})+x_{i}^{\ast
}(1/\psi _{i}^{\ast }(\theta )-1/\psi _{i}^{\ast }(\tilde{\theta}))| \\
& =n^{-1}\sum_{i=1}^{n}\left\vert \frac{v_{i}^{\ast \prime }}{\psi
_{i}^{\ast }(\bar{\theta}_{i})}\left( 1-\varepsilon _{i}^{\ast }\frac{\psi
_{i}^{\ast }(\hat{\theta}_{n})}{\psi _{i}^{\ast }(\bar{\theta}_{i})}\right)
(\theta -\tilde{\theta})\right\vert \\
& \leq || \theta -\tilde{\theta}|| \underset{C_{n}^{\ast }}{%
\underbrace{Cn^{-1}\sum_{i=1}^{n}(1+\varepsilon _{i}^{\ast })}}
\end{align*}%
using that $||v_{i}^{\ast \prime }/\psi _{i}^{\ast }(\theta )||$ and ${\psi _{i}^{\ast } (\hat{\theta}_n )}/\psi _{i}^{\ast }(\theta )$  are uniformly bounded on $\Theta .$ 
For the second term, similarly 
\begin{equation*}
|B_{n}\left( \theta \right) |\leq || \theta -\tilde{\theta}%
||\underset{C_{n}}{\underbrace{Cn^{-1}\sum_{i=1}^{n}(1+\varepsilon
_{i})}}.
\end{equation*}%
This establishes (ii) since $\mathbb{E}^{\ast }[C_{n}^{\ast }+C_{n}]=4C+o_{%
\text{a.s.}}(1)$ as $n\rightarrow \infty $. \hfill $\square $

\bigskip

Next, let $\mathcal{S}_{n}^{\ast }(\theta )=\partial \mathcal{L}_{n}^{\ast
}(\theta )/\partial \theta $ denote the first derivative of the bootstrap
log-likelihood, $\mathcal{S}_{n}^{\ast }(\theta )=\sum_{i=1}^{n}\xi
_{i}^{\ast }(\theta ),$ where by definition%
\begin{equation*}
\xi _{i}^{\ast }(\theta )=\left( \frac{x_{i}^{\ast }}{\psi _{i}^{\ast
}(\theta )}-1\right) \frac{v_{i}^{\ast }}{\psi _{i}^{\ast }(\theta )},\text{
\ \ }v_{i}^{\ast }=(1,x_{i-1}^{\ast })^{\prime }
\end{equation*}%
Set $\mathcal{S}_{n}^{\ast }=\mathcal{S}_{n}^{\ast }(\hat{\theta}_{n})$ and
let $\xi _{i}^{\ast }=\xi _{i}^{\ast }(\hat{\theta}_{n})=\left( \varepsilon
_{i}^{\ast }-1\right) v_{i}^{\ast }/\psi _{i}^{\ast }(\hat{\theta}_{n}).$

\begin{lemma}
\label{lem-A5-CPR}Under the assumptions of Lemma \ref%
{lem-asym-normality-boot-of-estimator}, it holds as $n\rightarrow \infty $%
\begin{equation*}
n^{-1/2}\mathcal{S}_{n}^{\ast }\overset{d^{\ast }}{\rightarrow }_{\text{a.s.}%
}\mathcal{N}\left( 0,\Omega \right) ,
\end{equation*}%
with $\Omega =\mathbb{E}[v_{i}v_{i}^{\prime }/\psi _{i}^{2}]$ as in (\ref%
{eq:Omega}).
\end{lemma}

\noindent\textsc{Proof:} With $\mathcal{F}_{i}^{\ast}=\sigma\left(
x_{i}^{\ast},x_{i-1}^{\ast},...,x_{0}^{\ast}\right) $, it follows that $%
\mathbb{E}^{\ast}[\xi_{i}^{\ast}|\mathcal{F}_{i-1}^{\ast}]=0$, and thus $%
\{\xi_{i}^{\ast}\}_{i=1}^{n}$ is a martingale difference sequence under $%
\mathbb{P}^{\ast}$ with respect to $\{\mathcal{F}_{i}^{\ast}\}_{i=1}^{n}$. The result
holds by verifying that%
\begin{align*}
\text{(i)}\text{:\ } & n^{-1}\sum_{i=1}^{n}\mathbb{E}^{\ast}[(\gamma
^{\prime}\xi_{i}^{\ast})^{2}|\mathcal{F}_{i-1}^{\ast}]\overset{p^{\ast }}{%
\rightarrow}_{\text{a.s.}}\gamma^{\prime}\Omega\gamma \\
\text{(ii)}\text{:\ } & n^{-1}\sum_{i=1}^{n}\mathbb{E}^{\ast}[(\gamma
^{\prime}\xi_{i}^{\ast})^{2}\mathbb{I}(|\gamma^{\prime}\xi_{i}^{\ast}|\geq
n^{1/2}\delta)]\rightarrow_{\text{a.s.}}0,\text{ \ for any }\delta>0
\end{align*}
for any $\gamma\in\mathbb{R}^{2}$ as $n\rightarrow\infty$.\medskip

\emph{Establishing (i):} It holds $(\gamma ^{\prime }\xi _{i}^{\ast
})^{2}=\left( \varepsilon _{i}^{\ast }-1\right) ^{2}(\gamma ^{\prime
}v_{i}^{\ast })^{2}/{\psi _{i}^{\ast}(\hat{\theta}_n)^2}$ and hence%
\begin{equation*}
n^{-1}\sum_{i=1}^{n}\mathbb{E}^{\ast }[(\gamma ^{\prime }\xi _{i}^{\ast
})^{2}|\mathcal{F}_{i-1}^{\ast }]=\frac{\mathbb{V}^{\ast }[\varepsilon
_{i}^{\ast }]}{n}\sum_{i=1}^{n}(\gamma ^{\prime }v_{i}^{\ast }/\psi
_{i}^{\ast }(\hat{\theta}_{n}))^{2},
\end{equation*}%
with $\mathbb{V}^{\ast }[\varepsilon _{i}^{\ast }]\rightarrow _{\text{a.s.}%
}1 $ by Lemma \ref{lem-A11-CPR}. Observe%
\begin{equation*}
\frac{1}{n}\sum_{i=1}^{n}\frac{(\gamma ^{\prime }v_{i}^{\ast })^{2}}{\psi
_{i}^{\ast }(\hat{\theta}_{n})^{2}}=\frac{1}{n}\sum_{i=1}^{n}\frac{(\gamma
^{\prime }v_{i}^{\ast })^{2}}{\psi _{i}^{\ast }(\theta _{0})^{2}}+\frac{1}{n}%
\sum_{i=1}^{n}\left( \frac{(\gamma ^{\prime }v_{i}^{\ast })^{2}}{\psi
_{i}^{\ast }(\hat{\theta}_{n})^{2}}-\frac{(\gamma ^{\prime }v_{i}^{\ast
})^{2}}{\psi _{i}^{\ast }(\theta _{0})^{2}}\right) \overset{p^{\ast }}{%
\rightarrow }_{\text{a.s.}}\gamma ^{\prime }\Omega \gamma
\end{equation*}%
using Lemma \ref{recursive-bootstrap-LLN} for the first term. The
convergence for the second term follows by the mean value theorem, $\hat{%
\theta}_{n}\rightarrow _{\text{a.s.}}\theta _{0}$ and $\left\vert \frac{%
\gamma ^{\prime }v_{i}^{\ast }}{\psi _{i}^{\ast }(\theta )}\right\vert \leq
C $, uniformly in $\theta $.

\bigskip

\emph{Establishing (ii)}:\ We have%
\begin{align*}
n^{-1}\sum_{i=1}^{n}\mathbb{E}^{\ast }[(\gamma ^{\prime }\xi _{i}^{\ast
})^{2}\mathbb{I}(|\gamma ^{\prime }\xi _{i}^{\ast }|\geq n^{1/2}\delta)]
&\leq \delta ^{-2}\frac{\mathbb{E}^{\ast }[\left( \varepsilon _{i}^{\ast
}-1\right) ^{4}]}{n}n^{-1}\sum_{i=1}^{n}\mathbb{E}^{\ast }\left[ \left(
\gamma ^{\prime }v_{i}^{\ast }/\psi _{i}^{\ast }(\hat{\theta}_{n})\right)
^{4}\right] \\
& \leq \delta ^{-2}C\frac{\mathbb{E}^{\ast }[\left( \varepsilon _{i}^{\ast
}-1\right) ^{4}]}{n}\rightarrow _{\text{a.s.}}0
\end{align*}%
by Lemma \ref{lem-A11-CPR}.\hfill $\square \bigskip $

Let $\mathcal{I}_{n}^{\ast }(\theta )=-\partial ^{2}\mathcal{L}_{n}^{\ast
}(\theta )/\partial \theta \partial \theta ^{\prime }$ denote the negative
second derivative of the bootstrap log-likelihood, $\mathcal{I}_{n}^{\ast
}(\theta )=\sum_{i=1}^{n}\zeta _{i}^{\ast }(\theta ),$ with%
\begin{equation*}
\zeta _{i}^{\ast }(\theta )=\left( 2\frac{x_{i}^{\ast }}{\psi _{i}^{\ast
}(\theta )}-1\right) \frac{v_{i}^{\ast }v_{i}^{\ast \prime }}{(\psi
_{i}^{\ast }(\theta ))^{2}}.
\end{equation*}%
Define $\mathcal{I}_{n}^{\ast }=\mathcal{I}_{n}^{\ast }(\hat{\theta}_{n})$
and let $\zeta _{i}^{\ast }=\zeta _{i}^{\ast }(\hat{\theta}_{n}).$

\begin{lemma}
\label{lem-A7-CPR}Under the assumptions of Lemma \ref%
{lem-asym-normality-boot-of-estimator}, it holds%
\begin{equation*}
n^{-1}\mathcal{I}_{n}^{\ast }\overset{p^{\ast }}{\rightarrow }_{\text{a.s.}%
}\Omega
\end{equation*}%
as $n\rightarrow \infty $ with $\Omega =\mathbb{E}[v_{i}v_{i}^{\prime }/\psi
_{i}^{2}]$ as in (\ref{eq:Omega}).
\end{lemma}

\noindent \textsc{Proof:} As $x_{i}^{\ast }/\psi _{i}^{\ast }(\hat{\theta}%
_{n})=\varepsilon _{i}^{\ast }$, we have%
\begin{equation*}
\frac{1}{n}\mathcal{I}_{n}^{\ast }=\frac{1}{n}\sum_{i=1}^{n}\frac{%
v_{i}^{\ast }v_{i}^{\ast \prime }}{\psi _{i}^{\ast }(\hat{\theta}_{n})^{2}}+%
\frac{2}{n}\sum_{i=1}^{n}\left( \varepsilon _{i}^{\ast }-1\right) \frac{%
v_{i}^{\ast }v_{i}^{\ast \prime }}{\psi _{i}^{\ast }(\hat{\theta}_{n})^{2}}.
\end{equation*}%
As shown above, 
\begin{equation*}
n^{-1}\sum_{i=1}^{n}\frac{v_{i}^{\ast }v_{i}^{\ast \prime }}{\psi _{i}^{\ast
}(\hat{\theta}_{n})^{2}}\overset{p^{\ast }}{\rightarrow }_{\text{a.s.}%
}\Omega .
\end{equation*}%
Similarly, with 
\begin{equation*}
D_{n}^{\ast }=n^{-1}\sum_{i=1}^{n}(\varepsilon _{i}^{\ast }-1)\frac{%
v_{i}^{\ast }v_{i}^{\ast \prime }}{\psi _{i}^{\ast }(\hat{\theta}_{n})^{2}},
\end{equation*}%
we obtain by the conditional Chebyshev inequality%
\begin{equation*}
\mathbb{P}^{\ast }\left( |\gamma ^{\prime }D_{n}^{\ast }\gamma |>\delta
\right) \leq \delta ^{-2}\frac{\mathbb{V}^{\ast }[\varepsilon _{i}^{\ast }]}{%
n}\frac{1}{n}\sum_{i=1}^{n}\mathbb{E}^{\ast }\left[\left( \gamma ^{\prime }
v_{i}^{\ast }v_{i}^{\ast \prime }\gamma /\psi _{i}^{\ast }(\hat{\theta}%
_{n})^{2} \right) ^{2}\right]\leq C\delta ^{-2}\frac{\mathbb{V}^{\ast
}[\varepsilon _{i}^{\ast }]}{n}\rightarrow _{\text{a.s.}}0,
\end{equation*}%
and we conclude $D_{n}^{\ast }\overset{p^{\ast }}{\rightarrow }_{\text{a.s.}%
}0$. \hfill $\square \bigskip $

In the lemma below, we denote $(\theta_{1},\theta_{2})=(\omega,\alpha)$, $%
(x_{1,i},x_{2,i})=(1,x_{i})$ and $(x_{1,i}^{\ast},x_{2,i}^{\ast})=(1,x_{i}^{%
\ast}).$

\begin{lemma}
\label{lem-A9-CPR}Under the assumptions of Lemma \ref%
{lem-asym-normality-boot-of-estimator}, it holds%
\begin{equation*}
\max_{l,j,k=1,2}\sup_{\theta \in \Theta }\left\vert \frac{1}{n}\frac{%
\partial ^{3}\mathcal{L}_{n}^{\ast }(\theta )}{\partial \theta _{l}\partial
\theta _{j}\partial \theta _{k}}\right\vert \leq C_{n}^{\ast },\text{ \ \ }%
C_{n}^{\ast }\overset{p^{\ast }}{\rightarrow }_{\text{a.s.}}C
\end{equation*}%
for some constant $C<\infty $.
\end{lemma}

\textsc{Proof:} For any $l,j,k=1,2$ and some constant $C_{0}<\infty$,%
\begin{equation*}
\left\vert \frac{1}{n}\frac{\partial ^{3}\mathcal{L}_{n}^{\ast }(\theta )}{%
\partial \theta _{l}\partial \theta _{j}\partial \theta _{k}}\right\vert
=\left\vert \frac{2}{n}\sum_{i=1}^{n}\left( 3\frac{x_{i}^{\ast }}{\psi
_{i}^{\ast }(\theta )}-1\right) \frac{x_{l,i-1}^{\ast }x_{j,i-1}^{\ast
}x_{k,i-1}^{\ast }}{\psi _{i}^{\ast 3}(\theta )}\right\vert \leq \frac{C_{0}%
}{n}\sum_{i=1}^{n}\left( 3\varepsilon _{i}^{\ast }+1\right) \overset{p^{\ast
}}{\rightarrow }_{\text{a.s.}}4C_{0}=:C\text{,}
\end{equation*}%
using $n^{-1}\sum_{i=1}^{n}\varepsilon _{i}^{\ast }\overset{p^{\ast }}{%
\rightarrow }_{\text{a.s.}}1$ by (the conditional) Chebyshev inequality as $%
\mathbb{V}^{\ast }[\varepsilon _{i}^{\ast }]=O_{\text{a.s.}}(1)$ as $%
n\rightarrow \infty $ by Lemma \ref{lem-A11-CPR}.\hfill $\square $

\subsection{Auxiliary bootstrap lemmas}

\begin{lemma}
\label{lem-A11-CPR}Under the assumptions of Theorem \ref{thm-bootstrap-main}%
, it holds for all $y\in \mathbb{R}$ that%
\begin{equation}
\mathbb{P}^{\ast }(\varepsilon _{i}^{\ast }\leq y)=\frac{1}{n(T)}%
\sum_{i=1}^{n(T)}\mathbb{I}(\hat{\varepsilon}_{i}^{s}\leq y)\rightarrow _{%
\text{a.s.}}\mathbb{P}(\varepsilon _{i}\leq y),
\label{lem-A11-CPR 1st result}
\end{equation}%
as $T\rightarrow \infty $. Furthermore, for any $p\geq 0$,
\begin{equation}
\mathbb{E}^{\ast }[(\varepsilon _{i}^{\ast })^{p}]=\tfrac{1}{n(T)}%
\sum_{i=1}^{n(T)}(\hat{\varepsilon}_{i}^{s})^{p}\rightarrow _{\text{a.s.}}%
\mathbb{E}[\varepsilon _{i}^{p}],  \label{lem-A11-CPR 2nd result}
\end{equation}%
as $T\rightarrow \infty $.
\end{lemma}

\noindent \textsc{Proof.} Define the ratio between the fitted residual and
the innovation as 
\begin{equation*}
q_{i,T}:=\frac{\hat{\varepsilon}_{i}}{\varepsilon _{i}}=\frac{\psi
_{i}(\theta _{0})}{\psi _{i}(\hat{\theta}_{n(T)})}=\frac{\omega _{0}+\alpha
_{0}x_{i-1}}{\hat{\omega}_{n(T)}+\hat{\alpha}_{n(T)}x_{i-1}}.
\end{equation*}%
For every $i\leq n(T)$, 
\begin{equation*}
|q_{i,T}-1|\leq \frac{|\omega _{0}-\hat{\omega}_{n(T)}|}{\omega _{L}}+\frac{%
|\alpha _{0}-\hat{\alpha}_{n(T)}|}{\alpha _{L}}\leq C\Vert \hat{\theta}%
_{n(T)}-\theta _{0}\Vert ,
\end{equation*}%
where the first inequality uses $\hat{\omega}_{n(T)}\geq \omega _{L}>0$ and $%
x/(\hat{\omega}_{n(T)}+\hat{\alpha}_{n(T)}x)\leq 1/\alpha _{L}$ for every $%
x\geq 0$. Consequently, with $d_{T}:=\sup_{1\leq i\leq n(T)}|q_{i,T}-1|$,
consistency of $\hat{\theta}_{n(T)}$ implies%
\begin{equation}
d_{T}\leq C\Vert \hat{\theta}_{n(T)}-\theta _{0}\Vert \rightarrow _{\text{%
a.s.}}0.  \label{eq-residual-ratio}
\end{equation}%
Next, using $\bar{\varepsilon}=n(T)^{-1}\sum_{i=1}^{n(T)}\hat{\varepsilon}%
_{i}$ and $\hat{\varepsilon}_{i}=q_{i,T}\varepsilon _{i}$, 
\begin{equation*}
\left\vert \bar{\varepsilon}-\frac{1}{n(T)}\sum_{i=1}^{n(T)}\varepsilon
_{i}\right\vert =\left\vert \frac{1}{n(T)}\sum_{i=1}^{n(T)}(q_{i,T}-1)%
\varepsilon _{i}\right\vert \leq d_{T}\frac{1}{n(T)}\sum_{i=1}^{n(T)}%
\varepsilon _{i}\rightarrow _{\text{a.s.}}0
\end{equation*}%
as by \citet[Theorem~2.1]{gut2009} $n(T)^{-1}\sum_{i=1}^{n(T)}\varepsilon
_{i}\rightarrow _{\text{a.s.}}\mathbb{E}[\varepsilon _{i}]=1$ since $%
n(T)\rightarrow _{\text{a.s.}}\infty $. This directly implies $\bar{%
\varepsilon}\rightarrow _{\text{a.s.}}1$.

Define next $r_{i,T}:=q_{i,T}/\bar{\varepsilon}$, $1\leq i\leq n(T)$, such
that $\hat{\varepsilon}_{i}^{s}=r_{i,T}\varepsilon _{i}$ and%
\begin{equation*}
\delta _{T}:=\sup_{1\leq i\leq n(T)}|r_{i,T}-1|=\sup_{1\leq i\leq
n(T)}\left\vert \frac{q_{i,T}-\bar{\varepsilon}}{\bar{\varepsilon}}%
\right\vert \leq \frac{d_{T}+|\bar{\varepsilon}-1|}{\bar{\varepsilon}}%
\rightarrow _{\text{a.s.}}0,
\end{equation*}%
by (\ref{eq-residual-ratio}) and $\bar{\varepsilon}\rightarrow _{\text{a.s.}%
}1$.

We now establish convergence of the bootstrap distribution function. Let 
\begin{equation*}
F_{n(T)}(z):=\frac{1}{n(T)}\sum_{i=1}^{n(T)}\mathbb{I}(\varepsilon _{i}\leq
z),\qquad \hat{F}_{T}(z):=\frac{1}{n(T)}\sum_{i=1}^{n(T)}\mathbb{I}(\hat{%
\varepsilon}_{i}^{s}\leq z).
\end{equation*}%
The Glivenko--Cantelli theorem and Lemma~\ref{lem-properties-of-NT} imply 
\begin{equation}
\sup_{z\in \mathbb{R}}|F_{n(T)}(z)-F_{\varepsilon }(z)|\rightarrow _{\text{%
a.s.}}0,  \label{eq-innovation-ecdf}
\end{equation}%
where $F_{\varepsilon }(z)=\mathbb{P}(\varepsilon _{i}\leq z)$. Fix $\eta
\in (0,1)$ such that $\delta _{T}\leq \eta $ for sufficiently large $T$ and $%
(1-\eta )\varepsilon _{i}\leq \hat{\varepsilon}_{i}^{s}\leq (1+\eta
)\varepsilon _{i}$. Hence, for every $y>0$, 
\begin{equation}
F_{n(T)}\left( y/(1+\eta )\right) \leq \hat{F}_{T}(y)\leq F_{n(T)}\left(
y/(1-\eta )\right) .  \label{eq-scaled-residual-cdf-bounds}
\end{equation}%
Using (\ref{eq-innovation-ecdf}) in (\ref{eq-scaled-residual-cdf-bounds})
yields, with probability one, 
\begin{equation*}
F_{\varepsilon }\left( y/(1+\eta )\right) \leq \liminf_{T\rightarrow \infty }%
\hat{F}_{T}(y)\leq \limsup_{T\rightarrow \infty }\hat{F}_{T}(y)\leq
F_{\varepsilon }\left( y/(1-\eta )\right) .
\end{equation*}%
By continuity of $F_{\varepsilon }$, letting $\eta \rightarrow 0$ proves $%
\hat{F}_{T}(y)\rightarrow _{\text{a.s.}}F_{\varepsilon }(y)$ for every $y>0$%
, and (\ref{lem-A11-CPR 1st result}) follows.

Consider (\ref{lem-A11-CPR 2nd result}). Since $r_{i,T}>0$ and $\delta
_{T}\rightarrow _{\text{a.s.}}0$, by the same arguments as above, $\delta
_{T,p}:=\sup_{1\leq i\leq n(T)}|r_{i,T}^{p}-1|$ satisfies $\delta
_{T,p}\rightarrow _{\text{a.s.}}0$. Using $(\hat{\varepsilon}%
_{i}^{s})^{p}=r_{i,T}^{p}\varepsilon _{i}^{p}$, it follows that 
\begin{align*}
\left\vert \frac{1}{n(T)}\sum_{i=1}^{n(T)}(\hat{\varepsilon}_{i}^{s})^{p}-%
\mathbb{E}[\varepsilon _{i}^{p}]\right\vert & \leq \left\vert \frac{1}{n(T)}%
\sum_{i=1}^{n(T)}(r_{i,T}^{p}-1)\varepsilon _{i}^{p}\right\vert +\left\vert 
\frac{1}{n(T)}\sum_{i=1}^{n(T)}\varepsilon _{i}^{p}-\mathbb{E}[\varepsilon
_{i}^{p}]\right\vert \\
& \leq \delta _{T,p}\frac{1}{n(T)}\sum_{i=1}^{n(T)}\varepsilon
_{i}^{p}+\left\vert \frac{1}{n(T)}\sum_{i=1}^{n(T)}\varepsilon _{i}^{p}-%
\mathbb{E}[\varepsilon _{i}^{p}]\right\vert \rightarrow _{\text{a.s.}}0
\end{align*}%
as $n(T)^{-1}\sum_{i=1}^{n(T)}\varepsilon _{i}^{p}\rightarrow _{\text{a.s.}}%
\mathbb{E}[\varepsilon _{i}^{p}]$ and $\delta _{T,p}\rightarrow _{\text{a.s.}%
}0$. Hence, using $\mathbb{E}^{\ast }[(\varepsilon _{i}^{\ast
})^{p}]=n(T)^{-1}\sum_{i=1}^{n(T)}(\hat{\varepsilon}_{i}^{s})^{p}$, (\ref%
{lem-A11-CPR 2nd result}) follows.\hfill $\square $

\bigskip

The following lemma applies, in particular, to $f(x)=\log (\omega +\alpha x)$
and $f(x)=((c+bx)/(\omega +\alpha x))^{k}$, where $\omega ,\alpha >0$, $%
c,b\in \mathbb{R}$, and $k\in \mathbb{N}$, which are used, e.g., in the
proof of Lemma \ref{lem-A2-CPR}. Both functions have bounded derivatives on $%
[0,\infty )$. The second function is bounded, while $|\log (\omega +\alpha
x)|=o(x^{p})$ as $x\rightarrow \infty $, for every $p>0$.

\begin{lemma}
\label{recursive-bootstrap-LLN}Assume that $f:[0,\infty)\rightarrow\mathbb{R}
$ is a Lipschitz continuous function satisfying 
\begin{equation}
|f(x)|\leq x^{p}+c(p),  \label{app:-dom}
\end{equation}
for any $p>0$, and with the constant $c(p)\in\lbrack0,\infty)$. Under the
assumptions of Theorem~\ref{thm-bootstrap-main}, it holds for the bootstrap
process $\{x_{i}^{\ast}\}_{i=1}^{n(T)}$ in (\ref{eq-RCD}), as $%
T\rightarrow\infty$,%
\begin{align}
& \frac{1}{n(T)}\sum_{i=1}^{n(T)}f(x_{i}^{\ast})-\mathbb{E}[f(x_{i})]\overset%
{p^{\ast}}{\rightarrow}_{\text{a.s.}}0  \label{Res1} \\
& \frac{1}{n(T)}\sum_{i=1}^{n(T)}\mathbb{E}^{\ast}[f(x_{i}^{\ast })]-\mathbb{%
E}[f(x_{i})]\rightarrow_{\text{a.s.}}0,  \label{Res2}
\end{align}
\end{lemma}

\noindent \textsc{Proof of Lemma \ref{recursive-bootstrap-LLN}:}\ Throughout
the proof, we use $s$ to denote a constant $s\in (0,1)$ for which $\mathbb{E}%
[(\alpha _{0}\varepsilon _{i})^{s}]<1$. Such an $s$ exists since $\mathbb{E}%
[\log (\alpha _{0}\varepsilon _{i})]=\lim_{s\rightarrow 0}(\mathbb{E}%
[(\alpha _{0}\varepsilon _{i})^{s}]-1)/s$ and $\mathbb{E}[\log (\alpha
_{0}\varepsilon _{i})]<0$ by assumption. Next, introduce some notation.\ Let 
\begin{equation}
x_{i}^{\ast }(\theta ;x)=(\omega +\alpha x_{i-1}^{\ast }(\theta
;x))\varepsilon _{i}^{\ast },\quad i\geq 1  \label{def-aux-boot-variables}
\end{equation}%
with $x_{0}^{\ast }(\theta ;x)=x$, such that $x_{i}^{\ast }=x_{i}^{\ast }(%
\hat{\theta}_{n(T)};x_{0})$. For the construction of a stationary bootstrap
process, extend the bootstrap innovations, conditionally on the data, to a
two-sided sequence $\{\varepsilon _{i}^{\ast }\}_{i\in \mathbb{Z}}$ of
i.i.d. variables with the empirical distribution of the scaled residuals.
The positive-index innovations in this sequence are those used in (\ref%
{def-aux-boot-variables}). Next, we write $x_{i,\infty }^{\ast }$ to denote%
\begin{equation*}
x_{i,\infty }^{\ast }=\omega _{0}\sum_{j=0}^{\infty }\alpha
_{0}^{j}\prod_{m=0}^{j}\varepsilon _{i-m}^{\ast },\quad i\in\mathbb{Z}.
\end{equation*}%
The process $\{x_{i,\infty }^{\ast }\}$ is well-defined for all $T$ large in
the sense that $\mathbb{E}^{\ast }[(x_{i,\infty }^{\ast })^{s}]=O_{\text{a.s.%
}}(1)$ as $T\rightarrow \infty $, as by Lemma \ref{lem-A11-CPR}, with
probability one, 
\begin{equation*}
\lim_{T\rightarrow \infty }\mathbb{E}^{\ast }[(\alpha _{0}\varepsilon
_{1}^{\ast })^{s}]=\mathbb{E}[(\alpha _{0}\varepsilon _{1})^{s}]<1.
\end{equation*}%
This implies, for some finite $T_{0}$, that $\mathbb{E}^{\ast }[(\alpha
_{0}\varepsilon _{1}^{\ast })^{s}]<\rho $, for some $0\leq \rho <1$, all $%
T\geq T_{0}$. The statement of the lemma holds as $T\rightarrow \infty $,
and we use $T\geq T_{0}$ throughout the proof. In particular,%
\begin{equation}
\mathbb{E}^{\ast }[(x_{i,\infty }^{\ast })^{s}]\leq (\omega _{0}/\alpha
_{0})^{s}\sum_{j=0}^{\infty }(\mathbb{E}^{\ast }[(\alpha _{0}\varepsilon
_{1}^{\ast })^{s}])^{j}\leq C.  \label{bound-on-boot-s-moment}
\end{equation}%
With this notation, decompose%
\begin{equation*}
\frac{1}{n(T)}\sum_{i=1}^{n(T)}f(x_{i}^{\ast })-\mathbb{E}[f(x_{1})]
=A_{1T}^{\ast }+A_{2T}^{\ast }+A_{3T}^{\ast }+A_{4T}^{\ast },
\end{equation*}%
where%
\begin{align*}
A_{1T}^{\ast }& :=\frac{1}{n(T)}\sum_{i=1}^{n(T)} \left( f(x_{i}^{\ast
})-f(x_{i}^{\ast }(\theta _{0},x_{0}))\right) , \\
A_{2T}^{\ast }& :=\frac{1}{n(T)}\sum_{i=1}^{n(T)} \left( f(x_{i}^{\ast
}(\theta _{0},x_{0}))-f(x_{i,\infty }^{\ast })\right) , \\
A_{3T}^{\ast }& :=\frac{1}{n(T)}\sum_{i=1}^{n(T)}f(x_{i,\infty }^{\ast }) -%
\mathbb{E}^{\ast }[f(x_{1,\infty }^{\ast })], \\
A_{4T}^{\ast }& :=\mathbb{E}^{\ast }[f(x_{1,\infty }^{\ast })] -\mathbb{E}%
[f(x_{1})].
\end{align*}%
The first two terms account for replacing the estimated parameter by $\theta
_{0}$ and removing the effect of initialization, respectively. The third is
a centred average from the stationary bootstrap process, while the fourth is
the difference between the bootstrap and original stationary means. By
stationarity, these means do not depend on the time index $i$.

Since $n(T)$ and $A_{4T}^{\ast }$ are fixed conditionally on the data, and $%
\mathbb{E}^{\ast }[A_{3T}^{\ast }]=0$, taking conditional expectations gives%
\begin{equation*}
\frac{1}{n(T)}\sum_{i=1}^{n(T)}\mathbb{E}^{\ast }[f(x_{i}^{\ast })] -\mathbb{%
E}[f(x_{1})] =\mathbb{E}^{\ast }[A_{1T}^{\ast }] +\mathbb{E}^{\ast
}[A_{2T}^{\ast }]+A_{4T}^{\ast }.
\end{equation*}%
To establish (\ref{Res1}) and (\ref{Res2}), we prove the following four
statements as $T\rightarrow\infty$:%
\begin{align*}
& \text{(i)}\text{:}\text{ }\frac{1}{n(T)}\sum_{i=1}^{n(T)}\mathbb{E}^{\ast
}[|f(x_{i}^{\ast })-f(x_{i}^{\ast }(\theta _{0},x_{0}))|]\rightarrow _{\text{%
a.s.}}0, \\
& \text{(ii)}\text{:}\text{ }\frac{1}{n(T)}\sum_{i=1}^{n(T)}\mathbb{E}^{\ast
}[|f(x_{i}^{\ast }(\theta _{0},x_{0}))-f(x_{i,\infty }^{\ast })|]\rightarrow
_{\text{a.s.}}0, \\
& \text{(iii)}\text{:}\text{ }A_{3T}^{\ast }\overset{p^{\ast }}{\rightarrow }%
_{\text{a.s.}}0, \\
& \text{(iv)}\text{:}\text{ }A_{4T}^{\ast }\rightarrow_{\text{a.s.}}0.
\end{align*}%
Indeed, by the triangle inequality, $\mathbb{E}^{\ast }[|A_{1T}^{\ast }|]$
and $\mathbb{E}^{\ast }[|A_{2T}^{\ast }|]$ are bounded by the averages in
(i) and (ii), respectively. Thus, for $j=1,2$ and every $\delta>0$,
conditional Markov's inequality gives%
\begin{equation*}
\mathbb{P}^{\ast }(|A_{jT}^{\ast }|>\delta) \leq\delta^{-1}\mathbb{E}^{\ast
}[|A_{jT}^{\ast }|] \rightarrow_{\text{a.s.}}0.
\end{equation*}%
Together with (iii) and (iv), this proves (\ref{Res1}) by the first
decomposition. Moreover,%
\begin{equation*}
|\mathbb{E}^{\ast }[A_{jT}^{\ast }]| \leq\mathbb{E}^{\ast }[|A_{jT}^{\ast
}|]\rightarrow_{\text{a.s.}}0, \qquad j=1,2.
\end{equation*}%
Using these limits and (iv) in the conditional-expectation decomposition
proves (\ref{Res2}).

\emph{Establishing (i):} It holds for any $k\in\mathbb{N}$ that%
\begin{align*}
|f(x)-f(y)|^{k} & =|f(x)-f(y)|^{k-s/2}|f(x)-f(y)|^{s/2} \\
& \leq C(|f(x)|+|f(y)|)^{k-s/2}|x-y|^{s/2}
\end{align*}
using Lipschitz continuity. Next since $|f(x)|\leq|x|^{p}+c_{p}$ for any $%
p>0 $, setting $p=(s/2)/(k-s/2)$ yields 
\begin{equation*}
(|f(x)|+|f(y)|)^{k-s/2}\leq(|x|^{p}+|y|^{p}+C)^{(s/2)/p}\leq
C(|x|^{s/2}+|y|^{s/2}+C),
\end{equation*}
using twice that $(|z|+|w|)^{m}\leq2^{m}(|z|^{m}+|w|^{m})$ for any $m\geq0$.
Thus%
\begin{equation*}
|f(x)-f(y)|^{k}\leq C(1+|x|^{s/2}+|y|^{s/2})|x-y|^{s/2}.
\end{equation*}
Using this, we obtain for any $i$%
\begin{align*}
\mathbb{E}^{\ast}[|f(x_{i}^{\ast})-f(x_{i}^{\ast}(\theta_{0},x_{0}))|^{k}] &
\leq C\mathbb{E}^{\ast}[(1+|x_{i}^{\ast}|^{s/2}+|x_{i}^{\ast}(%
\theta_{0},x_{0})|^{s/2})|x_{i}^{\ast}-x_{i}^{\ast}(\theta_{0},x_{0})|^{s/2}]
\\
& \leq C(\mathbb{E}^{\ast}[(1+|x_{i}^{\ast}|^{s/2}+|x_{i}^{\ast}(%
\theta_{0},x_{0})|^{s/2})^{2}])^{1/2} \\
& \cdot(\mathbb{E}^{\ast}[|x_{i}^{\ast}-x_{i}^{\ast}(%
\theta_{0},x_{0})|^{s}])^{1/2}
\end{align*}
by the Cauchy--Schwarz inequality. Using again $(|z|+|w|)^{m}%
\leq2^{m}(|z|^{m}+|w|^{m})$ for $m>0$,%
\begin{equation*}
\mathbb{E}^{\ast}[(1+|x_{i}^{\ast}|^{s/2}+|x_{i}^{\ast}(%
\theta_{0},x_{0})|^{s/2})^{2}]\leq C(1+\mathbb{E}^{\ast}[|x_{i}^{\ast}|^{s}]+%
\mathbb{E}^{\ast }[|x_{i}^{\ast}(\theta_{0},x_{0})|^{s}]).
\end{equation*}
By substitution we have%
\begin{equation}
x_{i}^{\ast}(\theta,x)=\omega\sum_{j=0}^{i-1}\alpha^{j}\prod_{k=0}^{j}%
\varepsilon_{i-k}^{\ast}+x\alpha^{i}\prod_{k=0}^{i-1}\varepsilon
_{i-k}^{\ast},  \label{eq-xstar-substi}
\end{equation}
and hence as $\mathbb{E}^{\ast}[(\hat{\alpha}_{n(T)}\varepsilon_{1}^{%
\ast})^{s}]\rightarrow_{\text{a.s.}}\mathbb{E}[(\alpha_{0}%
\varepsilon_{1})^{s}]$ as $T\rightarrow\infty$ by $\hat{\alpha}%
_{n(T)}\rightarrow_{\text{a.s.}}\alpha_{0}$ and Lemma \ref{lem-A11-CPR}, and 
$\mathbb{E}[(\alpha _{0}\varepsilon_{1})^{s}]<1$, it follows $\mathbb{E}%
^{\ast}[(\hat{\alpha }_{n(T)}\varepsilon_{1}^{\ast})^{s}]<\rho$ with
probability one for all $T$ sufficiently large. Therefore uniformly in $i$ 
\begin{equation*}
\mathbb{E}^{\ast}[|x_{i}^{\ast}|^{s}]\leq(\hat{\omega}_{n(T)}/\hat{\alpha }%
_{n(T)})^{s}\frac{1}{1-\rho}+x_{0}^{s}=O_{\text{a.s.}}(1)\text{,\quad as }%
T\rightarrow\infty.
\end{equation*}
The same argument applies to $\mathbb{E}^{\ast}[|x_{i}^{\ast}(%
\theta_{0},x_{0})|^{s}]$. Collecting terms,%
\begin{equation}
\mathbb{E}^{\ast}[|f(x_{i}^{\ast})-f(x_{i}^{\ast}(\theta_{0},x_{0}))|^{k}]%
\leq C_{T}(\mathbb{E}^{\ast}[|x_{i}^{\ast}-x_{i}^{\ast}(%
\theta_{0},x_{0})|^{s}])^{1/2}  \label{eq-key-ineq}
\end{equation}
for some $C_{T}=O_{\text{a.s.}}(1)$. Before we proceed on $%
|x_{i}^{\ast}-x_{i}^{\ast}(\theta_{0},x_{0})|$, notice that (\ref%
{eq-key-ineq}) is valid for any bootstrap variables $X^{\ast}$ and $Y^{\ast}$
satisfying $\mathbb{E}^{\ast}[(X^{\ast})^{s}],\mathbb{E}^{\ast}[(Y^{%
\ast})^{s}]=O_{\text{a.s.}}(1)$ as $T\rightarrow\infty$, i.e.%
\begin{equation}
\mathbb{E}^{\ast}[|f(X^{\ast})-f(Y^{\ast})|^{k}]\leq C_{T}(\mathbb{E}^{\ast
}[|X^{\ast}-Y^{\ast}|^{s}])^{1/2},  \label{eq-key-boot-ineq}
\end{equation}
which we will use repeatedly. Considering $|x_{i}^{\ast}-x_{i}^{\ast}(%
\theta_{0},x_{0})|$, it follows that%
\begin{align*}
|x_{i}^{\ast}-x_{i}^{\ast}(\theta_{0},x_{0})| & =|(\hat{\omega}_{n(T)}+\hat{%
\alpha}_{n(T)}x_{i-1}^{\ast})\varepsilon_{i}^{\ast}-(\omega_{0}+%
\alpha_{0}x_{i-1}^{\ast}(\theta_{0},x_{0}))\varepsilon_{i}^{\ast}| \\
& \leq|\hat{\omega}_{n(T)}-\omega_{0}|\varepsilon_{i}^{\ast}+|\hat{\alpha }%
_{n(T)}-\alpha_{0}|\varepsilon_{i}^{\ast}x_{i-1}^{\ast}+\alpha_{0}%
\varepsilon_{i}^{\ast}|x_{i-1}^{\ast}-x_{i-1}^{\ast}(\theta_{0},x_{0})|,
\end{align*}
and hence as $\varepsilon_{i}^{\ast}$ and $(x_{i-1}^{\ast},x_{i-1}^{\ast
}(\theta_{0},x_{0}))$ are conditionally independent, it holds%
\begin{equation*}
\mathbb{E}^{\ast}[|x_{i}^{\ast}-x_{i}^{\ast}(\theta_{0},x_{0})|^{s}]\leq
C_{T}(|\hat{\omega}_{n(T)}-\omega_{0}|^{s}+|\hat{\alpha}_{n(T)}-\alpha
_{0}|^{s})+\mathbb{E}^{\ast}[(\alpha_{0}\varepsilon_{i}^{\ast})^{s}]\mathbb{E%
}^{\ast}[|x_{i-1}^{\ast}-x_{i-1}^{\ast}(\theta_{0},x_{0})|^{s}],
\end{equation*}
using again that $\mathbb{E}^{\ast}[(x_{i-1}^{\ast})^{s}]$ is almost surely
bounded uniformly in $i$. Again, since $\mathbb{E}^{\ast}[(\alpha
_{0}\varepsilon_{i}^{\ast})^{s}]<\rho<1$, we have by recursion%
\begin{equation*}
\mathbb{E}^{\ast}[|x_{i}^{\ast}-x_{i}^{\ast}(\theta_{0},x_{0})|^{s}]\leq
C_{T}\frac{|\hat{\omega}_{n(T)}-\omega_{0}|^{s}+|\hat{\alpha}%
_{n(T)}-\alpha_{0}|^{s}}{1-\rho}\rightarrow_{\text{a.s.}}0,
\end{equation*}
uniformly in $i$. We conclude $\frac{1}{n(T)}\sum_{i=1}^{n(T)}\mathbb{E}%
^{\ast}[|f(x_{i}^{\ast})-f(x_{i}^{\ast}(\theta_{0},x_{0}))|^{k}]\rightarrow_{%
\text{a.s.}}0$ for any $k\in\mathbb{N}$. To prove (i), we only needed $k=1$,
but we will reuse some of the established results later.

\emph{Establishing (ii)}: We know from the above that $\mathbb{E}^{\ast
}[|x_{i}^{\ast }(\theta _{0},x_{0})|^{s}]\leq C_{T}$ for some $C_{T}=O_{%
\text{a.s.}}(1)$ and that $\mathbb{E}^{\ast }[|x_{i,\infty }^{\ast }|^{s}]$
has the same property. Hence using (\ref{eq-key-boot-ineq}), it holds%
\begin{equation*}
\mathbb{E}^{\ast }[|f(x_{i}^{\ast }(\theta _{0},x_{0}))-f(x_{i,\infty
}^{\ast })|]\leq C_{T}(\mathbb{E}^{\ast }[|x_{i,\infty }^{\ast }-x_{i}^{\ast
}(\theta _{0},x_{0})|^{s}])^{1/2}.
\end{equation*}%
Reusing the substitution in (\ref{eq-xstar-substi}), 
\begin{equation*}
|x_{i,\infty }^{\ast }-x_{i}^{\ast }(\theta _{0},x_{0})|\leq \omega
_{0}\sum_{j=i}^{\infty }\alpha _{0}^{j}\prod_{m=0}^{j}\varepsilon
_{i-m}^{\ast }+x_{0}\alpha _{0}^{i}\prod_{k=0}^{i-1}\varepsilon _{i-k}^{\ast
},
\end{equation*}%
and hence, as $\mathbb{E}^{\ast }[(\alpha _{0}\varepsilon _{1}^{\ast
})^{s}]<\rho <1$, it holds with probability one%
\begin{equation*}
\mathbb{E}^{\ast }[|x_{i}^{\ast }(\theta _{0},x_{0})-x_{i,\infty }^{\ast
}|^{s}]\leq \left\{ \left( \frac{\omega _{0}}{\alpha _{0}}\right) ^{s}\frac{%
\rho }{1-\rho }+x_{0}^{s}\right\} \rho ^{i}=:C_{0}\rho ^{i},
\end{equation*}%
Therefore, using the geometric-series bound,
\begin{align*}
\frac{1}{n(T)}\sum_{i=1}^{n(T)}\mathbb{E}^{\ast }[|f(x_{i}^{\ast }(\theta
_{0},x_{0}))-f(x_{i,\infty }^{\ast })|]&\leq \frac{C_{T}\sqrt{C_{0}}}{n(T)}%
\sum_{i=1}^{n(T)}\rho ^{i/2} \\
&\leq \frac{C_{T}\sqrt{C_{0}}}{n(T)}\frac{\sqrt{\rho}}{1-\sqrt{\rho}}
\rightarrow _{\text{a.s.}}0,
\end{align*}%
as $C_{T}=O_{\text{a.s.}}(1)$ and $n(T)\rightarrow \infty $ a.s. by Lemma~%
\ref{lem-properties-of-NT}.

\emph{Establishing (iii):} By Chebyshev's (conditional) inequality%
\begin{equation*}
\mathbb{P}^{\ast}\left( \left\vert \frac{1}{n(T)}\sum_{i=1}^{n(T)}f(x_{i,%
\infty}^{\ast})-\mathbb{E}^{\ast}[f(x_{i,\infty}^{\ast})]\right\vert
>\delta\right) \leq
\end{equation*}%
\begin{equation*}
\frac{1}{\left( \delta n(T)\right) ^{2}}\sum_{i=1}^{n(T)}\sum
_{j=1}^{n(T)}\left( \mathbb{E}^{\ast}[f(x_{j,\infty}^{\ast})f(x_{i,\infty}^{%
\ast})]-\mathbb{E}^{\ast}[f(x_{j,\infty}^{\ast})]\mathbb{E}%
^{\ast}[f(x_{i,\infty}^{\ast})]\right)
\end{equation*}
By stationarity of $x_{i,\infty}^{\ast}$ conditional on the data, it holds
for any $i,j$ that 
\begin{equation}
\mathbb{E}^{\ast}[f(x_{j,\infty}^{\ast})f(x_{i,\infty}^{\ast})]-\mathbb{E}%
^{\ast}[f(x_{j,\infty}^{\ast})]\mathbb{E}^{\ast}[f(x_{i,\infty}^{\ast })]=%
\mathbb{E}^{\ast}[f(x_{m,\infty}^{\ast})f(x_{0,\infty}^{\ast})]-\left( 
\mathbb{E}^{\ast}[f(x_{0,\infty}^{\ast})]\right) ^{2},
\label{eq-initial-relation-iii}
\end{equation}
with $m=|i-j|$. Observe that by definition, 
\begin{equation*}
x_{m,\infty}^{\ast}=\omega_{0}\sum_{j=0}^{\infty}\alpha_{0}^{j}%
\prod_{k=0}^{j}\varepsilon_{m-k}^{\ast}=\omega_{0}\sum_{j=0}^{m-1}%
\alpha_{0}^{j}\prod_{k=0}^{j}\varepsilon_{m-k}^{\ast}+x_{0,\infty}^{\ast}%
\alpha_{0}^{m}\prod_{k=0}^{m-1}\varepsilon_{m-k}^{\ast},
\end{equation*}
such that by the definition of $x_{i}^{\ast}(\theta,x)$ in (\ref%
{def-aux-boot-variables})%
\begin{equation*}
x_{m,\infty}^{\ast}=x_{m}^{\ast}(\theta_{0},0)+x_{0,\infty}^{\ast}\alpha
_{0}^{m}\prod_{k=0}^{m-1}\varepsilon_{m-k}^{\ast},
\end{equation*}
where, importantly, $x_{m}^{\ast}(\theta_{0},0)$ is independent of $%
x_{0,\infty}^{\ast}$ under $\mathbb{P}^{\ast}$. By simple rewriting,%
\begin{align*}
\mathbb{E}^{\ast}[f(x_{m,\infty}^{\ast})f(x_{0,\infty}^{\ast})] &=\mathbb{E}%
^{\ast}[f(x_{m}^{\ast}(\theta_{0},0))] \mathbb{E}^{\ast}[f(x_{0,\infty}^{%
\ast})] \\
&\quad+\mathbb{E}^{\ast}\!\left[ f(x_{0,\infty}^{\ast})
\{f(x_{m,\infty}^{\ast})-f(x_{m}^{\ast}(\theta_{0},0))\} \right].
\end{align*}
Inserting into (\ref{eq-initial-relation-iii}), we obtain%
\begin{equation*}
\left\vert \mathbb{E}^{\ast}[f(x_{m,\infty}^{\ast})f(x_{0,\infty}^{\ast
})]-\left( \mathbb{E}^{\ast}[f(x_{0,\infty}^{\ast})]\right) ^{2}\right\vert
\leq R_{1,m}+R_{2,m}
\end{equation*}
with%
\begin{align*}
R_{1,m} & =\mathbb{E}^{\ast}[\left\vert f(x_{0,\infty}^{\ast})\right\vert
]\left\vert \mathbb{E}^{\ast}[f(x_{m}^{\ast}(\theta_{0},0))]-\mathbb{E}%
^{\ast }[f(x_{0,\infty}^{\ast})]\right\vert \\
R_{2,m} & =\mathbb{E}^{\ast}\left[ \left\vert f(x_{0,\infty}^{\ast
})\right\vert \left\vert f(x_{m,\infty}^{\ast})-f(x_{m}^{\ast}(\theta
_{0},0))\right\vert \right] .
\end{align*}
Considering initially $R_{1,m}$ recall that $|f(x)|\leq x^{s}+c_{s}$, and $%
\mathbb{E}^{\ast}[(x_{0,\infty}^{\ast})^{s}]=O_{\text{a.s.}}(1)$, such that%
\begin{equation*}
R_{1,m}\leq C_{T}\left\vert \mathbb{E}^{\ast}[f(x_{m}^{\ast}(\theta
_{0},0))]-\mathbb{E}^{\ast}[f(x_{0,\infty}^{\ast})]\right\vert .
\end{equation*}
Next since $\mathbb{E}^{\ast}[f(x_{0,\infty}^{\ast})]=\mathbb{E}^{\ast
}[f(x_{m,\infty}^{\ast})]$, it holds%
\begin{equation*}
R_{1,m}\leq C_{T}\left\vert \mathbb{E}^{\ast}[f(x_{m}^{\ast}(\theta
_{0},0))]-\mathbb{E}^{\ast}[f(x_{m,\infty}^{\ast})]\right\vert \leq C_{T}%
\mathbb{E}^{\ast}[\left\vert
f(x_{m,\infty}^{\ast})-f(x_{m}^{\ast}(\theta_{0},0))\right\vert ].
\end{equation*}
Noting that $x_{m}^{\ast}(\theta_{0},0)\leq x_{m,\infty}^{\ast}$ and $%
\mathbb{E}^{\ast}[(x_{m,\infty}^{\ast})^{s}]=O_{\text{a.s.}}(1)$, we use (%
\ref{eq-key-boot-ineq}) to obtain the first inequality 
\begin{align*}
\mathbb{E}^{\ast}[\left\vert f(x_{m,\infty}^{\ast})-f(x_{m}^{\ast}(\theta
_{0},0))\right\vert ] & \leq C_{T}\mathbb{E}^{\ast}\left[ |x_{m,\infty
}^{\ast}-x_{m}^{\ast}(\theta_{0},0)|^{s}\right] ^{1/2} \\
& \leq C_{T}\rho^{m/2}
\end{align*}
The second inequality follows from 
\begin{equation*}
\mathbb{E}^{\ast}\!\left[|x_{m,\infty}^{\ast}-x_{m}^{\ast}(\theta_{0},0)|^{s}%
\right] = \mathbb{E}^{\ast}[(x_{0,\infty}^{\ast})^{s}] \left(\mathbb{E}%
^{\ast}[(\alpha_{0}\varepsilon_{1}^{\ast})^{s}]\right)^{m} \leq \mathbb{E}%
^{\ast}[(x_{0,\infty}^{\ast})^{s}]\rho^{m},
\end{equation*}
and from $\mathbb{E}^{\ast}[(x_{0,\infty}^{\ast})^{s}]=O_{\text{a.s.}}(1)$.
Hence $R_{1,m}\leq C_{T}\rho^{m/2}$.

For the $R_{2,m}$-term, we note by Cauchy-Schwarz that%
\begin{align*}
R_{2,m} & =\mathbb{E}^{\ast}\left[ \left\vert f(x_{0,\infty}^{\ast
})\right\vert \left\vert f(x_{m,\infty}^{\ast})-f(x_{m}^{\ast}(\theta
_{0},0))\right\vert \right] \\
& \leq\mathbb{E}^{\ast}[|f(x_{0,\infty}^{\ast})|^{2}]^{1/2}\mathbb{E}^{\ast
}[|f(x_{m,\infty}^{\ast})-f(x_{m}^{\ast}(\theta_{0},0))|^{2}]^{1/2}.
\end{align*}
Again, since $(|f(x)|)^{2}\leq(x^{s/2}+c_{s/2})^{2}\leq4(x^{s}+c_{s/2}^{2}) $%
, we have $\mathbb{E}^{\ast}[|f(x_{0,\infty}^{\ast})|^{2}]=O_{\text{a.s.}%
}(1) $ as $T\rightarrow\infty$.

Similarly, since (\ref{eq-key-boot-ineq}) holds for any $k$, repeating the
arguments from $R_{1,m}$ yields%
\begin{equation*}
R_{2,m}\leq C_{T}\rho^{m/4}
\end{equation*}
In summary, recalling $m=|i-j|$ we have shown that%
\begin{equation*}
\frac{1}{n(T)^{2}}\sum_{i=1}^{n(T)}\sum _{j=1}^{n(T)}\left\vert \mathbb{E}%
^{\ast}[f(x_{|i-j|,\infty}^{\ast })f(x_{0,\infty}^{\ast})]-\left( \mathbb{E}%
^{\ast}[f(x_{0,\infty}^{\ast })]\right) ^{2}\right\vert \leq C_{T}\frac{1}{%
n(T)^{2}}\sum_{i=1}^{n(T)}\sum_{j=1}^{n(T)}\rho^{|i-j|/4}\rightarrow _{\text{%
a.s.}}0
\end{equation*}
as $T\rightarrow\infty$, as desired.

\emph{Establishing (iv): }We establish (iv) by showing that for any $\delta
>0$%
\begin{equation*}
\lim \sup_{T\rightarrow \infty }|\mathbb{E}^{\ast }[f(x_{m,\infty }^{\ast
})]-\mathbb{E}[f(x_{m})]|\leq \delta
\end{equation*}%
with probability one. To show this, note that for any $m$%
\begin{equation*}
x_{m}^{\ast }(\theta _{0},0)=\omega _{0}\sum_{j=0}^{m-1}\alpha
_{0}^{j}\prod_{k=0}^{j}\varepsilon _{m-k}^{\ast }
\end{equation*}%
and introduce similarly the non-bootstrap analog%
\begin{equation*}
x_{m}(\theta _{0},0)=\omega _{0}\sum_{j=0}^{m-1}\alpha
_{0}^{j}\prod_{k=0}^{j}\varepsilon _{m-k}.
\end{equation*}%
Decomposing%
\begin{align*}
& |\mathbb{E}^{\ast }[f(x_{m,\infty }^{\ast })]-\mathbb{E}[f(x_{m})]| \\
& \leq \mathbb{E}^{\ast }\!\left[ |f(x_{m,\infty }^{\ast })-f(x_{m}^{\ast
}(\theta _{0},0))|\right] \\
& \quad +\left\vert \mathbb{E}^{\ast }[f(x_{m}^{\ast }(\theta _{0},0))]-%
\mathbb{E}[f(x_{m}(\theta _{0},0))]\right\vert +\mathbb{E}\!\left[
|f(x_{m}(\theta _{0},0))-f(x_{m})|\right] .
\end{align*}%
For the first term, we apply (\ref{eq-key-boot-ineq}) with $k=1$. Since $%
x_{m}^{\ast }(\theta _{0},0)\leq x_{m,\infty }^{\ast }$, it follows that 
\begin{align*}
\mathbb{E}^{\ast }[|f(x_{m,\infty }^{\ast })-f(x_{m}^{\ast }(\theta
_{0},0))|]& \leq C_{T}\mathbb{E}^{\ast }[|x_{m,\infty }^{\ast }-x_{m}^{\ast
}(\theta _{0},0)|^{s}]^{1/2} \\
& \leq C_{T}\rho ^{m/2}.
\end{align*}%
Here $C_{T}=O_{\text{a.s.}}(1)$ and, by the bound in (\ref%
{bound-on-boot-s-moment}), it can be chosen such that, for a non-random
constant $C<\infty $,%
\begin{equation*}
\limsup_{T\rightarrow \infty }C_{T}\leq C
\end{equation*}%
with probability one. This implies 
\begin{equation*}
\limsup_{T\rightarrow \infty }\mathbb{E}^{\ast }[|f(x_{m,\infty }^{\ast
})-f(x_{m}^{\ast }(\theta _{0},0))|]\leq \rho ^{m/2}C\leq \delta /3
\end{equation*}%
with probability one for all $m$ sufficiently large, say $m\geq m_{0}$, with 
$m_{0}$ deterministic. By completely analogous arguments it holds for the
third, non-bootstrap term%
\begin{equation*}
\mathbb{E}[|f(x_{m}(\theta _{0},0))-f(x_{m})|]\leq \delta /3
\end{equation*}%
for all $m$ sufficiently large, say $m\geq m_{1}$, again with $m_{1}$
non-random. Fixing $m=\max \{m_{0},m_{1}\}$, we now work to show for the
middle term%
\begin{equation*}
\limsup_{T\rightarrow \infty }|\mathbb{E}^{\ast }[f(x_{m}^{\ast }(\theta
_{0},0))]-\mathbb{E}[f(x_{m}(\theta _{0},0))]|\leq \delta /3
\end{equation*}%
to complete the proof. As a first step, observe that $x_{m}^{\ast }(\theta
_{0},0)$ is a continuous function of $(\varepsilon _{m}^{\ast },\varepsilon
_{m-1}^{\ast },\dots ,\varepsilon _{1}^{\ast })$ and similarly $x_{m}(\theta
_{0},0)$ is the \emph{same} function of $(\varepsilon _{m},\varepsilon
_{m-1},\dots ,\varepsilon _{1})$. Furthermore, $f$ is continuous, and hence%
\begin{equation*}
X^{\ast }=f(x_{m}^{\ast }(\theta _{0},0))
\end{equation*}%
is a continuous function of $(\varepsilon _{m}^{\ast },\varepsilon
_{m-1}^{\ast },\dots ,\varepsilon _{1}^{\ast })$, and $X=f(x_{m}(\theta
_{0},0))$ the same continuous function of $(\varepsilon _{m},\varepsilon
_{m-1},\dots ,\varepsilon _{1})$. Next, we establish the desired%
\begin{equation*}
\limsup_{T\rightarrow \infty }|\mathbb{E}^{\ast }[X^{\ast }]-\mathbb{E}%
[X]|\leq \delta /3.
\end{equation*}%
To do so, introduce%
\begin{align*}
t(x)& =x\mathbb{I}(|x|\leq M)+M(\mathbb{I}(x>M)-\mathbb{I}(x<-M)) \\
r(x)& =x\mathbb{I}(|x|>M)-M(\mathbb{I}(x>M)-\mathbb{I}(x<-M)),
\end{align*}%
for $M\geq 0,$ such that 
\begin{equation*}
\mathbb{E}^{\ast }[X^{\ast }]-\mathbb{E}[X]=\underset{T_{1}^{\ast }}{%
\underbrace{\mathbb{E}^{\ast }[t(X^{\ast })]-\mathbb{E}[t(X)]}}+\underset{%
T_{2}^{\ast }}{\underbrace{\mathbb{E}^{\ast }[r(X^{\ast })]-\mathbb{E}[r(X)]}%
},
\end{equation*}%
since $x=t(x)+r(x)$. For $T_{1}^{\ast }$, as $t\left( \cdot \right) $ is
continuous and bounded for any fixed $M$, 
\begin{equation}
\lim \sup_{T\rightarrow \infty }|\mathbb{E}^{\ast }[t(X^{\ast })]-\mathbb{E}%
[t(X)]|=0  \label{continuous-and-bounded-part}
\end{equation}%
with probability one, provided $X^{\ast }\overset{d^{\ast }}{\rightarrow }_{%
\text{a.s.}}X,$ as $T\rightarrow \infty $. As $X^{\ast }$ is a continuous
function of $(\varepsilon _{m}^{\ast },\varepsilon _{m-1}^{\ast },\dots
,\varepsilon _{1}^{\ast })$, and $X$ is the same function of $(\varepsilon
_{m},\varepsilon _{m-1},\dots ,\varepsilon _{1})$, this follows from 
\begin{equation*}
\mathbb{P}^{\ast }(\varepsilon _{m}^{\ast }\leq y_{m},\dots ,\varepsilon
_{1}^{\ast }\leq y_{1})\rightarrow _{\text{a.s.}}\mathbb{P}(\varepsilon
_{m}\leq y_{m},\dots ,\varepsilon _{1}\leq y_{1}),\text{\quad as }%
T\rightarrow \infty
\end{equation*}%
which as the $\varepsilon _{i}^{\ast }$'s are i.i.d. under $\mathbb{P}^{\ast }$%
, holds by the marginal convergence%
\begin{equation*}
\mathbb{P}^{\ast }(\varepsilon _{m}^{\ast }\leq y)\rightarrow _{\text{a.s.}}%
\mathbb{P}(\varepsilon _{m}\leq y),\text{\quad as }T\rightarrow \infty .
\end{equation*}%
using Lemma \ref{lem-A11-CPR}. Next, for $T_{2}^{\ast }$ note that%
\begin{equation*}
|r(x)|\leq (|x|+M)\mathbb{I}(|x|>M)
\end{equation*}%
and therefore%
\begin{align*}
|\mathbb{E}^{\ast }[r(X^{\ast })]-\mathbb{E}[r(X)]|& \leq \mathbb{E}^{\ast
}[(|X^{\ast }|+M)\mathbb{I}(|X^{\ast }|>M)]+\mathbb{E}[(|X|+M)\mathbb{I}%
(|X|>M)] \\
& \leq \frac{2(\mathbb{E}^{\ast }[|X^{\ast }|^{2}]+\mathbb{E}[|X|^{2}])}{M}.
\end{align*}%
Using $X^{\ast }=f(x_{m}^{\ast }(\theta _{0},0))$, and $|f(x)|^{2}\leq
(x^{s/2}+c_{s/2})^{2}\leq 4(x^{s}+c_{s/2}^{2})$, 
\begin{equation*}
\mathbb{E}^{\ast }[|X^{\ast }|^{2}]\leq 4(\mathbb{E}^{\ast }[(x_{m}^{\ast
}(\theta _{0},0))^{s}]+C)\leq 4(\mathbb{E}^{\ast }[(x_{m,\infty }^{\ast
})^{s}]+C)
\end{equation*}%
since $x_{m}^{\ast }(\theta _{0},0)\leq x_{m,\infty }^{\ast }$. In conclusion%
\begin{equation*}
\limsup_{T\rightarrow \infty }\mathbb{E}^{\ast }[|X^{\ast }|^{2}]\leq C
\end{equation*}%
with probability one. By similar arguments, $\mathbb{E}[|X|^{2}]\leq C$, and
we conclude 
\begin{equation*}
\lim \sup_{T\rightarrow \infty }|\mathbb{E}^{\ast }[r(X^{\ast })]-\mathbb{E}%
[r(X)]|\leq \frac{C}{M}\leq \delta /3
\end{equation*}%
with probability one for sufficiently large, fixed $M>0$. \hfill $\square $

\end{document}